\documentclass[]{interact}

\usepackage{epstopdf}
\usepackage[caption=false]{subfig}

\usepackage{pgfplots}
\pgfplotsset{compat=1.18}
\usepackage{pgfplotstable}
\usepgfplotslibrary{groupplots}

\usepackage{dsfont}

\usepackage{nicefrac}

\usepackage{hyperref}

\usepackage{siunitx}

\usepackage[numbers,sort&compress]{natbib}
\bibpunct[, ]{[}{]}{,}{n}{,}{,}
\renewcommand\bibfont{\fontsize{10}{12}\selectfont}

\theoremstyle{plain}
\newtheorem{theorem}{Theorem}[section]

\newtheorem{corollary}[theorem]{Corollary}
\newtheorem{proposition}[theorem]{Proposition}

\theoremstyle{definition}

\theoremstyle{remark}
\newtheorem{remark}{Remark}

\definecolor{color1}{HTML}{1f77b4}

\tikzset{every node/.prefix style={font=\footnotesize}}
\pgfplotsset{twoonpage/.style={height=200,width=300}}
   \pgfplotsset{proxystyle/.style={
      x tick label style={font=\tiny,rotate=90},
      grid = both,
      grid style=dashed,
      xtick = {0,  27,  54,  81, 108, 135, 162, 189},
      xticklabels={20-05 , 20-11, 21-05, 21-11, 22-06, 22-12, 23-06, 23-12},
      xmin = 0,
      xmax = 210
   }} 

\pgfplotsset{SimStyle/.style={
      x tick label style={font=\small},
      grid = both,
      grid style=dashed,
   }}

\begin{document}

\articletype{ARTICLE TEMPLATE}

\title{Linear short rate model with several delays}

\author{
\name{ \'Alvaro Guinea Juli\'a\textsuperscript{a}\thanks{\'Alvaro Guinea Juli\'a. Email: agjulia@icai.comillas.edu}  and Alet Roux\textsuperscript{b}\thanks{Alet Roux. Email: alet.roux@york.ac.uk}}
\affil{ \textsuperscript{a}Comillas Pontifical University ICAI, Madrid, 28015, Spain; \textsuperscript{b}University of York, Heslington, YO10 5DD, United Kingdom.}
}

\maketitle

\begin{abstract}
This paper introduces a short rate model in continuous time that adds one or more memory (delay) components to the Merton model \citep{merton1970dynamic,merton1973theory} or the Vasi\v{c}ek model \citep{vasicek1977equilibrium} for the short rate. The distribution of the short rate in this model is normal, with the mean depending on past values of the short rate, and a limiting distribution exists for certain values of the parameters. The zero coupon bond price is an affine function of the short rate, whose coefficients satisfy a system of delay differential equations. This system can be solved analytically, obtaining a closed formula. An analytical expression for the instantaneous forward rate is given: it satisfies the risk neutral dynamics of the Heath-Jarrow-Morton model. Formulae for both forward looking and backward looking caplets on overnight risk free rates are presented. Finally, the proposed model is calibrated against forward looking caplets on SONIA rates and the United States yield curve.
\end{abstract}


\section{Introduction}

In the real world, interest rates react to the decisions made by central banks. This means that there are periods in which interest rates are low and periods in which interest rates are high. This behavior can be observed, for example, in the United States three month treasury rate; see Figure \ref{fig:BondRate}(a) which uses data from Yahoo Finance (\href{https://finance.yahoo.com/quote/%5EIRX?p=%5EIRX}{finance.yahoo.com}). Furthermore, interest rates tend to change in response to inflation or economic growth (see for example the work of Benhabib \cite{Benhabib2004}  and Lilley and Rogoff \cite{lilley2019case}). Inflation and economic growth change could provoke a response in central banks that vary interest rates to increase economic growth and decrease inflation. Hence, it can be argued that the value of interest rates depend on past events. Evidence of memory effects in interest rates can be found in the work of Duan and Jacobs \cite{duan2001short}, Meade and Maier \cite{MememoryMeade2003}, and Baillie \cite{BAILLIE19965}. This can also be seen in practice: Figure~\ref{fig:BondRate}(b) shows the sample auto-correlation function of the daily United States three month treasury rate.

\begin{figure}
 \begin{center}
\begin{tikzpicture}
\pgfplotstableread[col sep=comma,]{BondPrices.csv} {\prices}
\pgfplotstableread[col sep=comma,]{Autocorrelation.csv} {\acf}
    \begin{groupplot}[group style={group name=my plots, group size=2 by 1, horizontal sep = 1cm}, width = 7.5cm, height=6cm]
        \nextgroupplot[proxystyle]
        \addplot[color1] table [x expr={\coordindex}, y={BondRate}] {\prices};      
        \nextgroupplot[SimStyle,xlabel={Lag (days)}]
        \addplot[color1] table [x expr={\coordindex}, y={acf} ] {\acf};
    \end{groupplot}
\node[align=center,anchor=north] at ([yshift=-12mm]my plots c1r1.south) {(a) Three month Treasury rate};
\node[align=center,anchor=north] at ([yshift=-12mm]my plots c2r1.south) {(b) Sample auto-correlation function};
\end{tikzpicture}
\end{center}
\caption{ Daily United States three month treasury rate, 1 May 2020---1 May 2024.}
\label{fig:BondRate}
\end{figure}

In this paper, we propose a short rate model that depends on its past values through delay terms, making it non-Markovian. It is a direct extension of the well-known Vasi\v{c}ek model \citep{vasicek1977equilibrium}, and hence also extends the Merton model \citep{merton1970dynamic,merton1973theory}. Our model leads to an analytical formula for the zero-coupon bond price, and it is well suited to calibration and the pricing of derivatives such as caplets. The model is effectively a continuous-time version of a simple autoregressive time series model---autoregressive time series having been used to successfully model interest rates \citep{HANSEN2021106302,Lanne2003}. As shown in Section \ref{sect:estimation}, the proposed model is able to capture the short-memory behavior of short rate time series.  Another characteristic of this model is that it can generate negative values for short rates, a realistic future in view of the fact that there have been periods with negative interest rates in recent years \citep{inhoffen2021low}.

There is a small but growing body of literature on financial models based on delay differential equations. Flore and Nappo \cite{flore2019feynman} proposed a delayed version of the Cox-Ingersoll-Ross model for the short rate, deriving zero coupon bond prices in terms of the solution of a system of delay differential equation (where the solution cannot be found explicitly). Coffie \cite{Coffie+2023+67+89,Coffie2024} presented several variations of the delayed A{\"\i}t-Sahalia short rate model with jumps. More generally, in the realm of equity options valuation,  Arriojas et al. \cite{arriojas2007delayed} introduced a delayed version of the renowned Black-Scholes-Merton model, and Kim, Kim and Jo \cite{kim_kim_jo_2022} extended this model to include jumps. Lee, Kim and Kim \cite{LEE20112909} proposed a delayed geometric Brownian motion with stochastic volatility. Kazmerchuk, Swishchuk and Wu \cite{KAZMERCHUK200769} incorporated a delay parameter into a local volatility model to price European options. Furthermore, Swishchuk and Xu \cite{swishchuk2011pricing} and Swishchuk and Vadori \cite{swishchuk2014smiling} constructed stochastic volatility models with delay to model variance swaps. More recently, G{\'o}mez-Valle and Mart{\'i}nez-Rodr{\'i}guez \cite{gomez2023estimating} introduced a delayed variation of geometric Brownian motion to price commodity futures.

Our model is based on a delayed version of the Ornstein-Uhlenbeck process with Gaussian noise. Delayed versions of this process have been explored in the literature in various contexts. K{\"u}chler and Mensch \cite{kuchler1992langevins} examined the delayed Ornstein-Uhlenbeck process in detail, determining the limiting distribution and proving the existence of a stationary solution. Mackey and Nechaeva \cite{mackey1995solution} studied the moment stability of linear stochastic delay differential equations. Basse-O'Connor et al. \cite{basse2020stochastic} proved the existence of strong solutions for the delayed Ornstein-Uhlenbeck process with L\'evy noise and compared it with ARMA time series models. Ott \cite{ott2006ornstein} investigated the stability properties of paths generated by the delayed Ornstein-Uhlenbeck process. 

In recent years, overnight rates have emerged as the new standard for interest rate benchmarks, replacing interbank offered rates (IBOR). Examples of these overnight rates include the Sterling Overnight Index Average (SONIA) in the United Kingdom, the Secured Overnight Financing Rate (SOFR) in the United States, and the Euro Short-Term Rate (€STR) in the European Union. In this paper, we price caplets on overnight rates, based on the framework proposed by Lyashenko and Mercurio \cite{lyashenko10looking}, which relies on the concepts of extended zero-coupon bonds and the extended forward measure \citep[Section 4.2.4]{andersen2010interest}. The structure of the model enables us to derive the extended zero-coupon bond and extended forward measure from the conventional definitions of the zero-coupon bond and forward measure. Lyashenko and Mercurio \cite{lyashenko10looking} provided a continuous approximation for the daily published overnight rates, which we use to define both the backward-looking forward rate and the forward-looking rate. We develop a formula in our model that is capable of pricing caplets on both backward and forward-looking rates. 

Classical short-rate models have been applied to price caplets on overnight rates. For instance, Rutkowski and Bickersteth \cite{rutkowski2021pricing} utilized the Vasi\v{c}ek model to price SOFR derivatives, while Turfus \cite{turfus2022caplet} derived a closed-form formula for pricing caplets on backward-looking rates using the Hull-White model. Furthermore, Fontana \cite{Fontana2023} derives explicit formulas for caplets on backward-looking rates on general affine models. Liu and Song \cite{Liu01112024} used continuous-time Markov chain to approximate short rate models and price caplets on risk-free rates. More recently, Fontana, Grbac, and Schmidt \cite{Fontana2024} proposed a Hull-White model with scheduled jumps at regular intervals to price caplets on backward-looking rates. The authors showed empirical evidence of these scheduled jumps in risk-free rates. These jumps are also considered by Fang, Yeh, He and Lin \cite{FANG2024102392}, Brace, Gellert and Schl\"ogl \cite{Brace2024} and Schl\"ogl, Skov and Skovmand \cite{SKOV2024}. It would be possible to add scheduled jumps to the model proposed in \eqref{eq:r}; however, this is outside of the scope of this paper, and is left as future work.

Regarding derivative valuation, Arriojas et al. \cite{arriojas2007delayed} showed that delayed models can reproduce volatility smiles and skews. The term structure of implied volatilities given by caplets behaves differently for short and long maturities \citep{brigo2006interest}. We show in the numerical section that introducing a delay parameter allows the model to capture the differences between caplets with short and long maturities, in a way that the classical models do not.

The paper is structured as follows. In Section \ref{Sect:IntroModel}, we introduce the model and derive an analytical solution to the stochastic delay differential equation. In Section \ref{sec:laplacetransform}, we calculate the joint Laplace transform of the integrated short rate and the short rate itself. This result enables us to determine both the price of zero-coupon bonds and the distribution of the short rate. Section \ref{Sect:ZCB} is devoted to the presentation of a closed-form formula for the zero-coupon bond price with. In Section \ref{Sect:Dist}, we analyze the distribution of the short rate, and provide conditions for the existence a limiting distribution. An analytical expression for the instantaneous forward rate is derived in Section \ref{sect:forward}, and then used to obtain the implied short rates from past years; see Section \ref{sect:yield}. In Section \ref{sec:caplets}, we derive an analytical formula for pricing caplets on overnight risk-free rates. Estimation and calibration against real market data are presented in Section \ref{sec:num_exp_pract_app}. Finally, Appendix \ref{sec:change-of-measure} discusses a structure-preserving change of measure.

\section{Model with delay} \label{Sect:IntroModel}

Let $(\Omega,\mathcal{F},\mathds{Q},\left(\mathcal{F}_t\right)_{t\geq 0})$ be a filtered probability space. Let us assume that the short rate $r$ follows the stochastic delay differential equation
\begin{equation}
dr_t =  \left(a(t) + b r_t + \sum_{j=1}^N c_j r_{t-\tau_j}\right)dt + \sigma(t) dW_t \label{eq:r}
\end{equation} 
for all $t\ge0$, where $N\in\mathbb{N}$, $b,c_1,c_2,\ldots,c_N\in\mathds{R}$, $W=(W_t)_{t\geq0}$ is a Brownian motion with respect to $\left(\mathcal{F}_t\right)_{t\geq 0}$, $a:[0,\infty]\to \mathds{R}$, $\sigma:[0,\infty]\to (0,\infty)$ are two continuous functions, and $\tau_N >\tau_{N-1}> \ldots> \tau_1>0$. The initial condition is
\begin{equation}
 r_t = \phi(t) \text{ for }  t\in (-\infty,0], \label{eq:phi}
\end{equation}
where $\phi:(-\infty,0]\rightarrow\mathds{R}$ is a deterministic integrable function. In most practical applications one would only use values of $\phi$ for $t\in[-\tau_N,0]$.

We will take $\mathds{Q}$ to be a risk-neutral probability. It is possible to make a structure-preserving change of measure from the real-world probability $\mathds{P}$; see Appendix \ref{sec:change-of-measure} for details.

This model generalizes the Merton model (which is the special case $N=1$, $b=c_1=0$) as well as the Vasi\v{c}ek model (the special case with $b<0$ and $N=1$, $c_1=0$) in a number of different ways, depending on the values of the parameters $N$, $b$ and $c_1,\ldots,c_N$. For example, taking $a(t)=a \in \mathds{R}$, $N=1$, $b=0$, $c_1<0$ and $\sigma(t) = \sigma >0$ leads to
\[
dr_t =  \left(a + c_1r_{t-\tau_1}\right)dt + \sigma dW_t,
\]
which is a Vasi\v{c}ek model with a single delay term, with mean reversion speed of $\frac{1}{\lvert c_1\rvert}$ and long-term mean $\frac{a}{\lvert c_1\rvert}$. In its most general form (when $b<0$ and $c_1\neq0$), it can be interpreted as a Vasi\v{c}ek model with speed of mean reversion $\frac{1}{\lvert b\rvert}$, and short term memory, in the sense that the usual Vasi\v{c}ek long-term mean term $\frac{a}{\lvert b\rvert}$ is replaced by $\frac{a}{\lvert b\rvert}+\frac{c}{\lvert b\rvert}r_{t-\tau_1}$. 

We will show below that the stochastic differential equation \eqref{eq:r} has a unique strong solution, and it is possible to derive an explicit analytical formula for it. The explicit solution of equation \eqref{eq:r} uses the solution of the deterministic delay differential equation with initial value
\begin{equation}
\left.
\begin{aligned}
    x'(t) & = bx(t) + \sum_{j=1}^N c_j x(t-\tau_j) & & \text{ when } t>0 \\
    x(t) & = g(t) & & \text{ when } t\in [-\infty,0],
    \end{aligned}
    \right\} \label{eq:deter}
\end{equation} where $g:[-\infty,0] \to \mathds{R}$ is a continuous deterministic function with at most exponential growth, so that the Laplace transforms exist. The following result makes use of the function $R:\mathds{R}\rightarrow\mathds{R}$ defined as 
\begin{equation} \label{eq:R}
 R(t) = 
 \begin{cases}
 \displaystyle \sum_{n=0}^{\lfloor\frac{t}{\tau_1}\rfloor}  \sum_{|\alpha|= n} \frac{c^\alpha}{\alpha !} \left(t- \langle \alpha,\tau \rangle\right)^n e^{b\left(t-\langle \alpha,\tau \rangle\right)} H\left(t- \langle \alpha,\tau \rangle\right)   &\text{ when  }  t\ge0,\\
 0  &\text{ when  }t\in[-\infty,0),
 \end{cases} 
 \end{equation} where $H=\mathds{1}_{[0,\infty)}$ is the Heaviside function and the second summation in \eqref{eq:R} is expressed using multi-index notation with $\alpha=(\alpha_1,\alpha_2,\ldots,\alpha_N)$, $|\alpha| = \sum_{i=1}^N \alpha_i $, $c^\alpha= \prod_{i=1}^N c_i^{\alpha_i}$, $\alpha!= \prod_{i=1}^N \alpha_i!$, $ \langle \alpha,\tau \rangle = \sum_{i=1}^N \alpha_i \tau_i$, and $\alpha_i\geq 0$ for every $i=1,\ldots,N$.

\begin{proposition}\label{prop:DeterDelay}
    The unique solution of \eqref{eq:deter} is
    \[ x(t) = x(0) R(t) + \sum_{j=1}^N c_j \int_{0}^{\tau_j} R(t-s) g(s-\tau_j) ds \text{ for } t\geq 0,\] where $R$ is defined in  \eqref{eq:R}.
\end{proposition}
\begin{proof} Existence and uniqueness come from direct application of the method of steps; see Remark 3.3 in \cite{smith2011introduction}. To obtain the explicit solution of \eqref{eq:deter}, we make used of the Laplace transform, which, for suitable $f:[0,\infty)\rightarrow\mathds{C}$ and $s\in\mathds{C}$, is defined as
 \[
  L_f(s)  = \int_0^\infty f(u) e^{-s u} du.
 \]The initial value problem \eqref{eq:deter} can be expressed as
\begin{align}
       x'(t) 
       & =  bx(t) + \sum_{j=1}^N c_j g(t-\tau_j) \mathds{1}_{[0,\tau_j)}(t) +  \sum_{j=1}^N c_j x(t-\tau_j) H(t-\tau_j). \label{eq:deter2}
\end{align} Let us define $L_x$ and $L_{g_j}$ as the Laplace transforms of $x(.)$ and $g(.-\tau_j)\mathds{1}_{[0,\tau_j)}(.)$ respectively. By properties of the Laplace transform \citep[see, for example][Theorems 2.1, 2.4, Appendix B]{Dyke2014}, equation \eqref{eq:deter2} can be written as
\begin{align}
       sL_x(s) - x(0)
       & =  bL_x(s) + \sum_{j=1}^N c_j L_{g_j}(s) +  \sum_{j=1}^N c_je^{-\tau_j s} L_x(s).  \label{eq:deter3}
\end{align} Rearranging equation \eqref{eq:deter3} and using the formula for the sum of a geometric series,  we obtain 
\begin{align}
    L_x(s) & = \frac{x(0) + \sum_{j=1}^N c_j L_{g_j}(s)}{s-b-\sum_{j=1}^N c_j e^{-\tau_j s}} \nonumber\\
    & = \frac{x(0) + \sum_{j=1}^N c_j L_{g_j}(s)}{s-b}\frac{1}{1 -\frac{\sum_{j=1}^N c_j e^{-\tau_j s}}{s-b}} \nonumber \\
    & =  \frac{x(0) + \sum_{j=1}^N c_j L_{g_j}(s)}{s-b} \sum_{n=0}^\infty  \frac{\left(\sum_{j=1}^N c_j e^{-\tau_j s}\right)^n}{(s-b)^n} \nonumber\\
    & = \left( x(0) + \sum_{j=1}^N c_j L_{g_j}(s) \right)  \sum_{n=0}^\infty \sum_{|\alpha|= n} \frac{n! }{\alpha!} \frac{c^\alpha e^{- \langle \alpha,\tau \rangle s}}{(s-b)^{n+1}} \nonumber \\
    & =  x(0)  \sum_{n=0}^\infty \sum_{|\alpha|= n} \frac{n! }{\alpha!} \frac{c^\alpha e^{- \langle \alpha,\tau \rangle s}}{(s-b)^{n+1}} + \sum_{j=1}^N \sum_{n=0}^\infty \sum_{|\alpha|= n}  c_j \frac{n! }{\alpha!} \frac{c^\alpha e^{- \langle \alpha,\tau \rangle s}}{(s-b)^{n+1}}  L_{g_j}(s),\label{eq:Laplace_end}
\end{align} where the multinomial theorem is used in the second to last equality. The results in \eqref{eq:Laplace_end} are well defined as long as $s\neq b$ and $\left| \frac{\sum_{j=1}^N c_j e^{-\tau_j s}}{s-b}  \right|<1$.

By properties of the inverse Laplace transform \citep[see, for example][Example 1.2, Theorems 2.4 and 3.2]{Dyke2014}, we have that
\begin{align}
& \mathcal{L}^{-1} \left[  \frac{ n! e^{- \langle \alpha,\tau \rangle s}}{(s-b)^{n+1}} \right] (t)  =  \left(t- \langle \alpha,\tau \rangle\right)^n e^{b\left(t-\langle \alpha,\tau \rangle\right)} H\left(t- \langle \alpha,\tau \rangle\right), \label{eq:Lap1}\\
& \mathcal{L}^{-1} \left[  \frac{ n! e^{- \langle \alpha,\tau \rangle s}}{(s-b)^{n+1}} L_{g_j}(s) \right] (t)  \nonumber\\
& \qquad= \int_0^t   \left(t-s- \langle \alpha,\tau \rangle\right)^n e^{b\left(t-s-\langle \alpha,\tau \rangle\right)} H\left(t-s- \langle \alpha,\tau \rangle\right) g(s-\tau_j) \mathds{1}_{[0,\tau_j)}(s) ds.\label{eq:Lap2}
\end{align}
Equations \eqref{eq:Lap1}-\eqref{eq:Lap2} allow us to obtain the desired result when applying the inverse Laplace transform to equation \eqref{eq:Laplace_end}.
\end{proof}

From Proposition \eqref{prop:DeterDelay}, it follows that the function $R$ in  \eqref{eq:R} satisfies the initial value problem
\begin{equation}
\left.\begin{aligned}
    R'(t) & = bR(t) + \sum_{j=1}^N c_j R(t-\tau_j)  &&\text{when } t>0, \\
    R(t) & = \mathds{1}_{\{0\}}(t) &&\text{when } t\in [-\infty,0].
    \end{aligned}\right\} \label{eq:Rdeter}
\end{equation}  Finally we can get the analytical solution of equation \eqref{eq:r}. 

\begin{proposition} \label{prop:sol:SDE}
The strong solution $(r_t)_{t\ge0}$ of the stochastic differential equation \eqref{eq:r} with initial condition \eqref{eq:phi} is given by
\begin{multline} \label{eq:2.5}
 r_t  = R(t)r_0 + \int_0^t a(s)R(t-s)ds + \sum_{j=1}^N c_j\int_{-\tau_j}^0 R(t-s-\tau_j)\phi(s)ds\\
 + \int_0^t \sigma(s) R(t-s)dW_s 
\end{multline}
for all $t\ge0$, where $R$ is defined in \eqref{eq:R}. 
\end{proposition}

\begin{proof}
As in the proof of Proposition \ref{prop:DeterDelay}, existence and uniqueness come from applying the method of steps \cite[p.~158]{mao1997}. We will show that the process \eqref{eq:2.5} satisfies equation \eqref{eq:r}. We use the stochastic Fubini's Theorem, which allows us to express the stochastic integral in \eqref{eq:2.5} as \begin{align}
\int_0^t \sigma(s) R(t-s)dW_s & =  \int_0^t \sigma(s) \left[\int_0^{t-s} R'(u) du + R(0) \right] dW_s\nonumber\\
&= \int_{0}^{t}\sigma(s) \left[ \int_{0}^{t} R'(u)\,\mathbf{1}_{\{u \le t - s\}} \,du \right] dW_{s} + \int_0^t \sigma(s) dW_s \nonumber\\
&= \int_{0}^{t} \left[ \int_{0}^{t} \sigma(s) \mathbf{1}_{\{s \le t-u\}}\, dW_{s} \right] R'(u)\,du  + \int_0^t \sigma(s) dW_s \nonumber \\
&= \int_{0}^{t} R'(u)\left[ \int_{0}^{t - u} \sigma(s) dW_{s} \right] du + \int_0^t \sigma(s) dW_s \nonumber \\
& = \int_{0}^{t} R'(t-u) \left[ \int_{0}^{u} \sigma(s) dW_{s} \right]du + \int_0^t \sigma(s) dW_s .\label{eq:StochInt1}
\end{align} Due to the fact that $R$ satisfies the initial value problem \eqref{eq:Rdeter}, we arrive at
\begin{align}
 & \int_{0}^{t} R'(t-u) \left[ \int_{0}^{u} \sigma(s) dW_{s} \right]du  \nonumber \\
 & \quad =  \int_{0}^{t} bR(t-u) \left[ \int_{0}^{u} \sigma(s) dW_{s} \right] du + \sum_{j=1}^N \int_{0}^{t} c_jR(t-u-\tau_j) \left[ \int_{0}^{u} \sigma(s) dW_{s} \right] du\nonumber\\
& \quad = \int_{0}^{t} bR(t-u) \left[ \int_{0}^{u} \sigma(s) dW_{s} \right] du + \sum_{j=1}^N  \int_{0}^{t-\tau_j} c_jR(t-u-\tau_j) \left[ \int_{0}^{u} \sigma(s) dW_{s} \right] du, \label{eq:IntStoch}
\end{align} where the last equality comes from the fact that $R(t) = 0$ when $t< 0$.
This implies that the process $r$ in \eqref{eq:2.5} can be expressed as 
\begin{equation}\label{eq:decomp}
    r_t =  f(t) + \int_0^t \sigma(s) dW_s \text{ for } t\geq 0,
\end{equation} where $f$ is defined as
\begin{multline*}
    f(t) = R(t)r_0  +
    \int_0^t a(s) R(t-s)ds  +  \sum_{j=1}^N c_j \int_{-\tau_j}^0 R(t-s-\tau_j)\phi(s)ds \\
    + \int_{0}^{t} bR(t-s) \left[ \int_{0}^{s} \sigma(u) dW_{u} \right] ds + \sum_{j=1}^N c_j  \int_{0}^{t-\tau_j} R(t-s-\tau_j) \left[ \int_{0}^{s} \sigma(u) dW_{u} \right] ds
\end{multline*}
for all $t\ge0$. Applying the Leibniz rule, we arrive at
\begin{align}
    f'(t) & = R'(t)r_0  + a(t)R(0) +\int_0^t a(s) R'(t-s) ds\nonumber\\
    & \quad + \sum_{j=1}^N c_j \int_{-\tau_j}^0 R'(t-s-\tau_j)\phi(s)ds  \nonumber \\
    & \quad  +  bR(0)\int_{0}^{t} \sigma(u) dW_{u} + \int_{0}^{t} bR'(t-s) \left[ \int_{0}^{s} \sigma(u) dW_{u} \right] ds  \nonumber \\
    & \quad  + \sum_{j=1}^N c_jR(0) \int_{0}^{t-\tau_j} \sigma(u) dW_{u}  + \sum_{j=1}^N c_j \int_{0}^{t-\tau_j}  R'(t-s-\tau_j) \left[ \int_{0}^{s} \sigma(u) dW_{u} \right] ds .\label{eq:f1}
\end{align} Equation \eqref{eq:StochInt1} together with the fact that $R(0)=1$, gives us
\begin{align}
&bR(0)\int_{0}^{t} \sigma(u) dW_{u} + \int_{0}^{t} bR'(t-s) \left[ \int_{0}^{s} \sigma(u) dW_{u} \right] ds  = b \int_0^t \sigma(s) R(t-s)dW_s \label{eq:b1}\\
& c_jR(0) \int_{0}^{t-\tau_j} \sigma(u) dW_{u} 
 + \int_{0}^{t-\tau_j} c_j R'(t-s-\tau_j) \left[ \int_{0}^{s} \sigma(u) dW_{u} \right] ds\nonumber\\
 & = c_j \int_0^{t-\tau_j} \sigma(s) R(t-s-\tau_j)dW_s.\label{eq:c1}
\end{align} Equations \eqref{eq:b1}--\eqref{eq:c1} and the initial value problem \eqref{eq:Rdeter} allow us to write equation \eqref{eq:f1} as
\begin{align}
    f'(t) & = br_0 R(t) + r_0 \sum_{j=1}^N c_j R(t-\tau_j) \nonumber\\
    & \qquad +  a(t) + b\int_0^t a(s) R(t-s) ds  + \sum_{j=1}^N c_j \int_0^{t-\tau_j} a(s) R(t-s-\tau_j) ds\nonumber\\
    & \qquad +  b \sum_{j=1}^N c_j   \int_{-\tau_j}^0 R(t-s-\tau_j)\phi(s)ds \nonumber\\
    & \qquad + \sum_{j=1}^N  c_j  \sum_{k=1}^N c_k \int_{-\tau_j}^0 R(t-s-\tau_j-\tau_k)\phi(s)ds  \nonumber \\
    & \qquad +  b \int_0^{t} \sigma(s) R(t-s)dW_s +\sum_{j=1}^N  c_j \int_0^{t-\tau_j} \sigma(s) R(t-s-\tau_j)dW_s . \label{eq:fPrima1} 
\end{align} Grouping the terms multiplied by $b$ and $c_j$ in equation \eqref{eq:fPrima1}, we obtain that
\begin{align}
    f'(t) & =  a(t) + br_t + \sum_{j=1}^N  c_j r_{t-\tau_j}.\label{eq:fPrima2} 
\end{align} Equations \eqref{eq:fPrima2} and \eqref{eq:decomp} show that the process $r$ in \eqref{eq:2.5} satisfies the stochastic delay differential equation \eqref{eq:r}.

\end{proof}

\section{Laplace transform} \label{sec:laplacetransform}

In this section we derive an explicit formula for the exponential-affine transform in \eqref{eq:transform} below. Introducing auxiliary processes $\gamma_j$ for $j=1,\ldots,N$, following an idea of Flore and Nappo \cite{flore2019feynman}, allows us to obtain an affine formula for the conditional expectation of the transform. Explicit formulae for bond prices and the characteristic function of the short rate will be derived in due course as special cases of this formula.

\begin{theorem} \label{th:transform}
For any $T>t\ge0$, $z\in\mathds{C}$ and $d_0,d_1\in\mathds{R}$, we have
\begin{equation}
E_\mathds{Q}\left(e^{zr_T+\int_t^T(d_0+d_1r_s)ds}\middle|\mathcal{F}_t\right) = e^{A(T-t) + D(T-t)r_t + \sum_{j=1}^N c_j\int_{t-\tau_j}^t D(T-u-\tau_j)r_u du} \label{eq:transform}
\end{equation}
where $A:[0,\infty)\rightarrow\mathds{R}$ and $D:\mathds{R}\rightarrow\mathds{R}$ satisfy the system of delay differential equations
\begin{align}
   D'(\ell) &= bD(\ell) + \sum_{j=1}^N c_jD(\ell-\tau_j) + d_1, \label{eq:Ricatti-D} \\
   A'(\ell) &= a(T-\ell)D(\ell) + \tfrac{1}{2}\sigma^2(T-\ell)D^2(\ell) + d_0 \label{eq:Ricatti-A}
\end{align}
for all $\ell>0$, with initial values $A(0)=0$ and $D(\ell)=zH(\ell)$ for all $\ell\in(-\infty,0]$, where $H$ is the Heaviside function.
\end{theorem}

\begin{proof}
  Taking any $T>0$, and taking as given a deterministic (and integrable) function $\Gamma_j:[0,T]\rightarrow\mathds{R}$ that will be chosen below, define an auxiliary process $\gamma_j=(\gamma_{t,j})_{t\in[0,T]}$ as
  \[
    \gamma_{t,j}
    = \int_{t-\tau_j}^t \Gamma_j(s) r_s \mathds{1}_{[-\tau_j,T-\tau_j]}(s)ds \text{ for all } t\in [0,T] \text{ and  for all } j=1,\ldots,N.  
  \]
  Note that $\gamma_{T,j}=0$ and, for all $t\in [0,T]$,
  \begin{align*}
         \gamma_{t,j}  & =  \int_{-\tau_j}^t \Gamma_j(s) r_s \mathds{1}_{[-\tau_j,T-\tau_j]} (s) ds -  \int_{0}^{t} \Gamma_j(s-\tau_j) r_{s-\tau_j} ds\\
 & = \gamma_{0,j} + \int_{0}^t \left(\Gamma_j(s) r_s \mathds{1}_{[0,T-\tau_j]} (s) - \Gamma_j(s-\tau_j) r_{s-\tau_j}\right) ds.
\end{align*}

  Define now
  \[
   \psi_t = E_\mathds{Q}\left(e^{zr_T+\int_t^T(d_0+d_1r_s)ds}\middle|\mathcal{F}_t\right) \text{for all }t\in[0,T].
  \]
  We claim that
\begin{equation}\label{Ansatzeq}
 \psi_t = e^{A(t) + D(t) r_t + \sum_{j=1}^N \gamma_{t,j} } \text{ for all } t\in [0,T],
\end{equation} 
where $A,D:[0,T]\rightarrow\mathds{R}$ are two deterministic functions with $A(T) = 0$ and $D(T) = z$. Application of the It\^o formula in \eqref{Ansatzeq} gives
\begin{align}
 \frac{d\psi_t}{\psi_t} 
 &= \left(A'(t) + a(t)D(t) + \tfrac{1}{2}\sigma^2(t)D^2(t)\right)dt \nonumber \\
 &\qquad + \left(D'(t)+bD(t)+ \sum_{j=1}^N \Gamma_j(t)\mathds{1}_{[0,T-\tau_j]}(t)\right)r_t dt \nonumber \\
 &\qquad + \left(\sum_{j=1}^N\left( c_j D(t) - \Gamma_j(t-\tau_j)\right) r_{t-\tau_j}\right)dt + \sigma(t) D(t)dW_t. \label{eq:SDE-gives-gamma}
\end{align}
The function $\Gamma_j$ is now chosen so as to allow $r_{t-\tau_j}$ to be eliminated from \eqref{eq:SDE-gives-gamma}. Noting that the value of $\gamma_j$ does not depend on $\Gamma_j(t)$ for any $t>T-\tau_j$, we simply choose $\Gamma_j$ so as to be continuous and constant for such $t$. In summary, we obtain
\begin{equation} \label{eq:Gamma}
 \Gamma_j(t) =
 \begin{cases}
  c_jD(t+\tau_j) & \text{if } t\in[-\tau_j,T-\tau_j], \\
  z & \text{if } t\in (T-\tau_j,T].
 \end{cases} \text{ for } j=1,\ldots,N.
\end{equation}

Define now another auxiliary process $y=(y_t)_{t\in[0,T]}$ as
  \[
   y_t = e^{\int_0^t(d_0+d_1r_s)ds} \text{ for all }t\ge0.
  \]
  It follows that
  \[
   y_t\psi_t = E_\mathds{Q}\left(y_T\psi_T\middle|\mathcal{F}_t\right) \text{ for all }t\in[0,T], 
  \]
  in other words, $(y_t\psi_t)_{t\in[0,T]}$ is a martingale. The It\^o product rule together with \eqref{eq:SDE-gives-gamma} and \eqref{eq:Gamma} gives that
  \begin{multline*}
   \frac{d(y_t\psi_t) }{y_t\psi_t} 
   = \left(A'(t) + a(t)D(t) + \tfrac{1}{2}\sigma^2(t)D^2(t) + d_0\right) dt \\
   + \left(D'(t)+bD(t) + \sum_{j=1}^N c_jD(t+\tau_j)\mathds{1}_{[0,T-\tau_j]}(t) + d_1\right)r_tdt + \sigma(t) D(t)dW_t.
  \end{multline*}
  The martingale property of $(y_t\psi_t)_{t\in[0,T]}$ then implies that $A$ and $D$ satisfies the system of differential equations
  \begin{align*}
   D'(t) &= -bD(t) - \sum_{j=1}^N c_j D(t+\tau_j )\mathds{1}_{[0,T-\tau_j]}(t) - d_1, \\
   A'(t) &= -a(t)D(t) -\tfrac{1}{2}\sigma^2(t)D^2(t) -d_0.
  \end{align*}
  By abuse of notation and the transformation $\ell=T-t$ we obtain \eqref{eq:Ricatti-D}--\eqref{eq:Ricatti-A} and finally~\eqref{eq:transform}.
\end{proof}

The key task now is to solve the delay differential equation \eqref{eq:Ricatti-D}; the following result provides an analytical solution. Once \eqref{eq:Ricatti-D} has been solved, the function $A$ can be obtained by integrating \eqref{eq:Ricatti-A}.

\begin{theorem} \label{th:dde}
The solution to the delay differential equation \eqref{eq:Ricatti-D} with initial condition $D(\ell)=zH(\ell)$ for all $\ell\in(-\infty,0]$ is given by 
\begin{align}
 D(\ell) & = z + d_1\sum_{n=0}^{\left\lfloor \frac{\ell}{\tau_1}\right\rfloor} \sum_{|\alpha|= n}  \frac{c^\alpha}{(n+1)\alpha! } (\ell-\langle \alpha,\tau \rangle)^{n+1} H(\ell - \langle \alpha,\tau \rangle) \nonumber \\
  & \qquad + z\sum_{j=1}^N\sum_{n=1}^{\left\lfloor \frac{\ell}{\tau_1}\right\rfloor} \sum_{|\alpha|= n} \frac{ c_j c^\alpha}{(n+1) \alpha !} (\ell-\langle \alpha,\tau \rangle - \tau_j)^{n+1} H(\ell - \langle \alpha,\tau \rangle -\tau_j)  \label{eq:D-without-b}
 \end{align}
for all $\ell\ge0$ when $b=0$, and when $b\neq0$ the function $D$ satisfies
\begin{align}
        D(\ell) & = z + (d_1+zb)\sum_{n=0}^{\left\lfloor \frac{\ell}{\tau_1}\right\rfloor} \sum_{|\alpha|= n}  \frac{c^\alpha}{\alpha!} D_{n,\alpha}(\ell-\langle \alpha,\tau \rangle) \nonumber\\
    & \qquad + z \sum_{j=1}^N \sum_{n=1}^{\left\lfloor \frac{\ell}{\tau_1}\right\rfloor}  \sum_{|\alpha|= n} \frac{c_j c^\alpha}{\alpha!} D_{n,\alpha}( \ell -\langle \alpha,\tau \rangle-\tau_j) \label{eq:D-with-b}
\end{align} for all $\ell\ge0$, where
\begin{align}
     D_{n,\alpha}(\ell) & =  \frac{(-1)^n n!}{b^{n+1}} \left(e^{b \ell }\sum_{r=0}^n\frac{(-1)^{-r}}{r!}b^r(\ell)^r - 1\right) H(\ell).\label{eq:Dnalpha}  
\end{align}
\end{theorem}

\begin{proof} Existence and uniqueness of equation \eqref{eq:Ricatti-D} come from the method of steps; see Remark 3.3 in \cite{smith2011introduction}. We use a number of elementary properties of the Laplace transform \citep[see, for example][p.~7, Example 1.2, Theorems 2.1, 2.4, Appendix B]{Dyke2014}, and take as given a value of $s$ for which all expressions below are well defined, i.e.~we require $s\notin\{0,b\}$ and $\left\lvert \frac{ce^{-\tau s}}{s-b}\right\rvert < 1$. 
 Observing that \eqref{eq:Ricatti-D} is equivalent to
 \[
  D'(\ell) = bD(\ell) + \sum_{j=1}^N c_j D(\ell-\tau_j)H(\ell-\tau_j) + d_1 \text{ for all }\ell>0,
 \]
 it follows that
 \[
  sL_D(s)-z =  L_{D'}(s)  = bL_D(s) + \sum_{j=1}^N c_j e^{-\tau_j s} L_D(s) + \frac{d_1}{s}.
 \]
 After rearrangement, this becomes
 \begin{align*}
       L_{D'}(s)  & = \frac{d_1+zb+z \sum_{j=1}^N c_j e^{-\tau_j s}}{s-b- \sum_{j=1}^N c_j e^{-\tau_j s}}\\
       & = \frac{1}{s-b}\left(d_1+zb+z  \sum_{j=1}^N c_j e^{-\tau_j s} \right)\frac{1}{1-\frac{ \sum_{j=1}^N c_j e^{-\tau_j s} }{s-b}}. 
 \end{align*}
  Using the formula for the sum of a geometric series, we obtain
  \begin{align}
  L_{D'}(s) &= \frac{1}{s-b}\left(d_1+zb+z \sum_{j=1}^N c_j e^{-\tau_j s}\right)\sum_{n=0}^\infty\left(\frac{\sum_{j=1}^N c_j e^{-\tau_j s}}{s-b}\right)^n \nonumber \\
  &= (d_1+zb)\sum_{n=0}^\infty \sum_{|\alpha|= n} \frac{n! }{\alpha!} \frac{c^\alpha e^{- \langle \alpha,\tau \rangle s}}{(s-b)^{n+1}} \nonumber\\
  & \qquad + z  \sum_{j=1}^N \sum_{n=0}^\infty \sum_{|\alpha|= n}  c_j \frac{n! }{\alpha!} \frac{c^\alpha e^{- (\langle \alpha,\tau \rangle +\tau_j) s}}{(s-b)^{n+1}}, \label{eq:Laplace-intermediate}
 \end{align} where we have used the multinomial theorem in the last equality. By properties of the inverse Laplace transform, we have that 
 \[
 \mathcal{L}^{-1}\left( \frac{n!e^{-h s}}{(s-b)^{n+1}} \right)   = (\ell-h)^ne^{b(\ell-h)}H(\ell-h)
 \]
 for all $n\in\mathds{N}_0$ and $h\ge0$, inverting the Laplace transform in \eqref{eq:Laplace-intermediate} gives
 \begin{multline}\label{eq:DerD1}
  D'(\ell)
  = (d_1+zb) \sum_{n=0}^\infty \sum_{|\alpha|= n} \frac{c^\alpha}{\alpha !} (\ell -  \langle \alpha,\tau \rangle )^ne^{b(\ell- \langle \alpha,\tau \rangle)}H(\ell- \langle \alpha,\tau \rangle) \\
  + z \sum_{j=1}^N \sum_{n=0}^\infty \sum_{|\alpha|= n} c_j \frac{c^\alpha}{\alpha !} (\ell-(\langle \alpha,\tau \rangle + \tau_j))^{n}e^{b(\ell-(\langle \alpha,\tau \rangle + \tau_j))}H(\ell-(\langle \alpha,\tau \rangle + \tau_j))
  \end{multline} for all $\ell\ge0$.
  This can be integrated directly, and
  \begin{multline} \label{eq:Laplace-intermediate-integrated}
   D(\ell) = z + (d_1+zb)\sum_{n=0}^{\left\lfloor \frac{\ell}{\tau_1}\right\rfloor}  \sum_{|\alpha|= n} \frac{c^\alpha}{\alpha!} \int_{\langle \alpha,\tau \rangle}^{\ell\vee \langle \alpha,\tau \rangle}(u- \langle \alpha,\tau \rangle)^ne^{b(u- \langle \alpha,\tau \rangle)}du \\
  + z \sum_{j=1}^N  \sum_{n=1}^{\left\lfloor \frac{\ell}{\tau_1}\right\rfloor}  \sum_{|\alpha|= n} \frac{c_j c^\alpha}{\alpha!}  \int_{\langle \alpha,\tau \rangle + \tau_j}^{\ell \vee (\langle \alpha,\tau \rangle + \tau_j)}(u-(\langle \alpha,\tau \rangle + \tau_j))^{n}e^{b(u-(\langle \alpha,\tau \rangle + \tau_j))}du \text{ for all } \ell\ge0.
  \end{multline}
  
We distinguish between two cases in \eqref{eq:Laplace-intermediate-integrated}, depending on the value of $b$. Take first $b\neq0$. For each $n\in\mathds{N}$ such that $\ell \ge k$ with \[k\in\{ \langle \alpha,\tau \rangle, \langle \alpha,\tau \rangle +\tau_1,\ldots, \langle \alpha,\tau \rangle +\tau_N \},\] a change of variable and integration by parts gives that
  \begin{align*}
   \int_{k}^\ell (u-k)^n e^{b(u-k)}du
   &= \frac{1}{b^{n+1}}\int_0^{b(\ell-k)} v^n e^v dv \\
   &= \frac{n!(-1)^n}{b^{n+1}}\left(e^{b(\ell-k)}\sum_{r=0}^n\frac{(-1)^{-r}}{r!}b^r(\ell-k)^r - 1\right).
  \end{align*}
  This formula also applies trivially when $n=0$, and therefore \eqref{eq:D-with-b} holds true.
  
  Similarly, if $b=0$, then the integrals in \eqref{eq:Laplace-intermediate-integrated} simplify significantly and we can obtain
  \begin{align*}
   \int_{k}^\ell (u-k)^n du 
   = \frac{1}{n+1}(\ell-k)^{n+1},
  \end{align*}
  which gives \eqref{eq:D-without-b} as claimed.
\end{proof}

\section{Zero coupon bonds}\label{Sect:ZCB}

Let us consider zero coupon bonds in this short rate model. The arbitrage-free price of a zero-coupon bond with maturity date $T>0$ is
\begin{align}\label{eq:BondPrce1}
    B(t,T) & =E_{\mathds{Q}}\left(\left. e^{-\int_t^T r_s ds }\right| \mathcal{F}_t \right) \text{ for all } t\in [0,T]
\end{align} \citep[(9.22)]{musiela2006martingale}. We have the following result.

\begin{proposition}\label{prop:BondPriceFormula}
 The arbitrage-free price of a zero-coupon bond with maturity date $T>0$ is given for all $t\in[0,T]$ by
 \begin{equation}
 B(t,T) = e^{A(T-t) + D(T-t)r_t + \sum_{j=1}^N c_j \int_{t-\tau_j}^t D(T-s-\tau_j)r_s ds}, \label{eq:bond}
 \end{equation}
 where the functions $A:[0,\infty)\rightarrow\mathds{R}$ and $D:\mathds{R}\rightarrow\mathds{R}$ are defined as $D(\ell) = 0$ for all $\ell\in(-\infty,0)$ and
 \begin{align}
  D(\ell) &=
  \begin{cases}
   - \displaystyle\sum_{n=0}^{\left\lfloor \frac{\ell}{\tau_1}\right\rfloor} \sum_{|\alpha|= n}  \frac{c^\alpha}{(n+1)\alpha! } (\ell-\langle \alpha,\tau \rangle)^{n+1} H(\ell - \langle \alpha,\tau \rangle) &\text{if }b=0, \\
   -\displaystyle \sum_{n=0}^{\left\lfloor \frac{\ell}{\tau_1}\right\rfloor} \sum_{|\alpha|= n}  \frac{c^\alpha}{\alpha!} D_{n,\alpha}(\ell) &\text{if }b\neq0,
  \end{cases} \label{eq:DSolBond}\\
  A(\ell) &= \int_0^\ell a(T-u) D(u)du + \tfrac{1}{2}\int_0^\ell \sigma^2(T-u) D^2(u)du
 \end{align}
 for all $\ell\ge0$, where $D_{n,\alpha}$ is defined as in equation \eqref{eq:Dnalpha}.
\end{proposition}

\begin{proof}
 This is a special case of Theorems \ref{th:transform} and \ref{th:dde} with $z=d_0=0$ and \mbox{$d_1=-1$}.
\end{proof}

The bond price formula \eqref{eq:bond} resembles the general exponential-affine structure of bond prices found in affine short rate models (see Musiela and Rutkowski \cite[Section 10.2.2]{musiela2006martingale}). 

Applying the It\^o formula in \eqref{eq:bond}, it is straightforward to establish that
\begin{align}\label{eq:BondRNP}
    dB(t,T) & =   r_t B(t,T) dt + \sigma(t) D(T-t) B(t,T) dW_t
\end{align} for all $t\in [0,T]$. This is a stochastic differential equation with deterministic volatility, and therefore it is possible to obtain closed formulae for options on bonds \citep[see, for example,][Proposition 11.3.1]{musiela2006martingale} in terms of $D$ and the parameters of the model, and by extension also for caps, floors and other derivative securities. This includes stock pricing models where the short rate follows \eqref{eq:r} under the risk neutral probability $\mathds{Q}$ \citep[Proposition 11.3.2]{musiela2006martingale}.

\section{Distribution of the short rate}\label{Sect:Dist}

The conditional characteristic function of the short rate, defined for all $T>t\ge0$ as
\begin{align}\label{eq:caractFunction}
    \phi^r_{T|t}(u) &=  E_{\mathds{Q}}\left(\left. e^{iu r_T}\right| \mathcal{F}_t \right) \text{ for all }  u\in \mathds{R}.
\end{align}
It be obtained as a special case of Theorem \ref{th:transform}, as follows.

\begin{proposition}
For any $T>t\ge0$ we have
\begin{align}
    \phi^r_{T|t}(u)& = e^{iu\left(\int_0^{T-t}a(T-s)R(s)ds + R(T-t)r_t + \sum_{j=1}^N c_j \int_{t-\tau_j}^t R(T-s-\tau_j)r_s ds\right)} \nonumber\\
 &\times   e^{-\frac{1}{2}u^2\int_0^{T-t} \sigma^2(T-s)R^2(s)ds} \label{eq:charfun}
\end{align}
for all $u\in \mathds{R}$, where $R$ is the function defined in \eqref{eq:R}.
\end{proposition}

\begin{proof}
Fix any $u\in\mathds{R}$. This corresponds to the special case of Theorems \ref{th:transform} and~\ref{th:dde} with $d_0=d_1=0$ and $z=iu$.

Direct calculation from \eqref{eq:D-without-b} and \eqref{eq:D-with-b} give
\[
 D(\ell) = iuR(\ell) \text{ for all } \ell\in \mathds{R}.
\]
It follows that  
 \[
 A(\ell) = iu\int_0^\ell a(T-s) R(s)ds - \tfrac{1}{2}u^2\int_0^\ell \sigma^2(T-s) R^2(s)ds \text{ for all }\ell\ge0\]
 by \eqref{eq:Ricatti-A} and the result follows immediately.
 \end{proof}
 
It can be inferred from Proposition \ref{prop:sol:SDE} that the (unconditional) distribution of $r_T$ is normal for all $T>0$. The form of \eqref{eq:charfun} shows that its conditional distribution is also normal, with the parameters being provided by the following corollary.

\begin{corollary}\label{cor:NormalDistr}
For any $0\le t<T$, the conditional distribution of $r_T$ given $\mathcal{F}_t$ is normal with mean
\begin{equation} \mu_{t,T} = \int_0^{T-t}a(T-u)R(u)du + R(T-t)r_t + \sum_{j=1}^N c_j \int_{t-\tau_j}^t R(T-u-\tau_j)r_u du \label{eq:cond-mean}
\end{equation}
and variance
\begin{equation} \sigma_{t,T}^2 = \int_0^{T-t} \sigma^2(T-u) R^2(u)du. \label{eq:cond-variance}
\end{equation}
\end{corollary}

The short rate being normally distributed is useful for simulating its values as well as estimating its parameters from historical data (using the maximum likelihood method, for example). Analytical (but complicated) formulae exist for all but one of the integrals in \eqref{eq:cond-mean}--\eqref{eq:cond-variance}; experience from numerical experiments suggests that numerical quadrature approximates the integrals (including the second integral in \eqref{eq:cond-mean}) very well.
\begin{remark}
By the tower property and equation \eqref{eq:charfun}, the characteristic function in \eqref{eq:caractFunction} can be expressed as \begin{align*}
     \phi^r_{T|t}(u) &=  E_{\mathds{Q}}\left(\left. e^{iu r_T}\right| \mathcal{F}_t \right) \\
     & = E_{\mathds{Q}}\left(\left. e^{iu r_T}\right| \sigma\left( \left\{r_{t-s}\right\}_{s=0}^{\tau_N} \right) \right)
\end{align*} for  $T>t\geq 0$  and 
 all  $u\in \mathds{R}$,  where $\sigma\left( \left\{r_{t-s}\right\}_{s=0}^{\tau_N}\right)$ is the sigma field generated by the random variables $r_{t-s}$ for $0\leq s\leq \tau_N$. This means that the conditional distribution of $r_T$ given $\mathcal{F}_t$, depends only on the values of the short rate from time $t-\tau_N$ up to time~$t$.
\end{remark}
We conclude this section by studying the stability of the process defined in \eqref{eq:r}. The stability of the solution of \eqref{eq:r} is related to the stability of the solution of the deterministic initial value problem \eqref{eq:deter}. The behavior of \eqref{eq:deter} is determined by the value of the roots of the characteristic function $h$, which is defined as
\[h(\lambda) =  \lambda -b -\sum_{j=1}^N c_je^{-\lambda \tau_j}  \]  for  $\lambda \in \mathds{C}$. A value of $\lambda$ that makes $h(\lambda) =0$ is called a characteristic root. Let us define the set $\Lambda$ as the set of all characteristic roots that is
\[ \Lambda = \left\{\lambda\in\mathds{C} : h(\lambda) =0\right\}.\] In fact, its stability is related to the supremum of the real parts of the characteristic roots, i.e.
\[ \nu_0 = \sup\{ \Re(\lambda) : \lambda\in \Lambda \}.\]
It is shown that the solution of \eqref{eq:deter} is asymptotically stable \citep[p. 39]{smith2011introduction}, if and only if $\nu_0 <0$ \citep[p. 533]{HALE1985533}. Next result from Hale, Infante and Fu-Shiang \cite{HALE1985533}, gives the conditions for which the solution of \eqref{eq:deter} is asymptotically stable.

\begin{theorem}\label{eq:StabilyDeter}\citep[Corollary 3.4]{HALE1985533}
The solution of the deterministic delay differential equation with initial value \eqref{eq:deter} is asymptotically stable for every value of the delay parameters $\tau_1,\ldots,\tau_N>0$ if and only if
\begin{equation}\label{eq:cond1}
    b + \sum_{j=1}^N c_j \neq 0, \quad |b| \geq  \sum_{j=1}^N |c_j| \text{ and } b<0.
\end{equation}
\end{theorem}

Finally, this last result allows us to give the conditions under which the unique solution of equation \eqref{eq:r} has a limiting distribution. 
\begin{proposition}\label{prop:stationary_distr}
Assume that the conditions in \eqref{eq:cond1} hold and that the functions 
\(a,\sigma : [0,\infty)\to\mathbb{R}\) are bounded and satisfy
\[
\lim_{t\to\infty} a(t) = a_\infty \in \mathbb{R},
\qquad
\lim_{t\to\infty} \sigma(t) = \sigma_\infty > 0.
\]
Then, for every choice of delay parameters \(\tau_1,\dots,\tau_N>0\), the process
\((r_t)_{t\ge 0}\) admits a limiting distribution: as \(T\to\infty\),
\(r_T\) converges in distribution to a normally distributed random variable with mean
\[
m_\infty = a_\infty \int_0^\infty R(u)\,du
\]
and variance
\[
v_\infty^2 = \sigma_\infty^2 \int_0^\infty R^2(u)\,du,
\]
\end{proposition}
\begin{proof}
    Since the function $R$ satisfies the initial value problem \eqref{eq:Rdeter}, we have that $R$ is asymptotically stable; see Theorem \ref{eq:StabilyDeter}. This implies that the function $R$ can be bounded by an exponential function,
    \[ |R(t)| \leq M e^{-\gamma t} \text{ for } t\geq 0\] for some $\gamma,M>0$  \citep[p~215, Corollary 6.1]{hale2006functional}. From last result and the dominated convergence Theorem,
    \begin{align*}
        \lim_{T\to \infty} \int_0^{T}a(T-u)R(u)du & =  \lim_{T\to \infty} \int_0^{\infty}a(T-u)R(u) \mathbf{1}_{[0,T]}(u) du\\
        & = \int_0^{\infty}a_\infty R(u) du.
    \end{align*} Similarly, we can obtain,
    \begin{align*}
    \lim_{T\to \infty} \int_0^{T}\sigma^2(T-u)R^2(u)du  &  = \sigma^2_\infty \int_0^{\infty} R^2(u)du.  
    \end{align*} For each $T>0$, the random variable $r_T$ is normal and  since the first two moments converge, it follows that the distribution of \(r_T\) converges weakly to a normal distribution with mean $m_\infty$ and variance $v^2_\infty$.
\end{proof}

\section{ Instantaneous forward rate} \label{sect:forward}

It is possible to compute the instantaneous forward rate in terms of the short rate. Furthermore, we will show that it satisfies the risk-neutral dynamics of the Heath-Jarrow-Morton model.

Let $D$ be defined as in \eqref{eq:DerD1} in Theorem \ref{th:dde} with $z=0$ and $d_1=-1$ (equivalently, equation \eqref{eq:DSolBond} in Proposition \ref{prop:BondPriceFormula}). Then it is clear from \eqref{eq:DerD1} that
\begin{align}\label{eq:derD}
    \frac{\partial D(\ell)}{\partial \ell} & = -R(\ell) \text{ for } \ell\in [0,T]. 
\end{align}
It then follows from \eqref{eq:bond}, in conjunction with
    \begin{align*}
        \int_0^{T-t} a(T-u) D(u)du & =  \int_t^{T} a(u) D(T-u)du, \\
      \int_0^{T-t} \sigma^2(T-u) D^2(u)du & =  \int_t^T \sigma^2(u) D^2(T-u)du,
    \end{align*}
that the instantaneous forward rate satisfies
\begin{align}
    f(t,T) & = - \frac{\partial}{\partial T} \ln B(t,T)\nonumber\\
    &= -a(T)D(0)  - \int_t^T a(u) D'(T-u) du \nonumber \\
    & \qquad - \frac{\sigma^2(T)}{2} D^2(0)  -  \int_t^T \sigma^2(u) D(T-u) D'(T-u) du  \nonumber \\
    & \qquad + R(T-t) r_t + \sum_{j=1}^N c_j  \int_{t-\tau_j}^t R(T-u-\tau_j)r_u du \nonumber \\
    & = \int_t^T a(u) R(T-u) du  +\int_t^T \sigma^2(u) D(T-u) R(T-u) du \nonumber\\
    & \qquad + R(T-t) r_t + \sum_{j=1}^N c_j \int_{t-\tau_j}^t R(T-u-\tau_j)r_u du.\label{eq:ForwardRates}
\end{align} 

The following result gives us a stochastic differential equation for the forward rate.

\begin{proposition}\label{prop:Forward}
The instantaneous forward rate defined in \eqref{eq:ForwardRates} satisfies the stochastic differential equation
 \[
 df(t,T) =  -\sigma(t)^2 D(T-t)R(T-t)  dt + \sigma(t) R(T-t)dW_t
 \]  for $t\in[0,T]$ with initial condition
 \begin{align*}
      f(0,T) & = \int_0^T a(u) R(T-u) du  +\int_0^T \sigma^2(u) D(T-u) R(T-u) du\\
      & \qquad + R(T) r_0 + \sum_{j=1}^N c_j  \int_{-\tau_j}^0 R(T-u-\tau_j)\phi(u) du.
 \end{align*}
\end{proposition}
\begin{proof}
Observe that the instantaneous forward rate in \eqref{eq:ForwardRates}  can be written as
\begin{align}
     f(t,T) &= \int_t^T a(u) R(T-u) du  +\int_t^T \sigma^2(u) D(T-u) R(T-u) du \nonumber\\ 
     & \qquad + R(T-t) r_t +\sum_{j=1}^N c_j \gamma_{t,j},\label{eq:ForwardRates2}
\end{align}
where the processes $\gamma_j=(\gamma_{t,j})_{t\in[0,T]}$ are defined as
  \[
    \gamma_{t,j}
    = \int_{t-\tau_j}^t  R(T-s-\tau_j) r_s \mathds{1}_{[-\tau_j,T-\tau_j]}(s)ds \text{ for all } t\in [0,T] \text{ and } j=1,2,\ldots,N.     
  \] Similar to the proof of Theorem \ref{th:transform}, we have that
  \begin{align*}
         \gamma_{t,j}  = \gamma_{0,j} + \int_{0}^t \left( R(T-s-\tau_j)  r_s - R(T-s) r_{s-\tau_j} \right) ds.
\end{align*} 

Equations \eqref{eq:derD} and \eqref{eq:ForwardRates2} allow us to write the forward rate in differential form as
\begin{align}
    df(t,T) & = \left[ -a(t)R(T-t) - \sigma(t)^2 D(T-t)R(T-t) -  R'(T-t) r_t\right] dt \nonumber \\
    & \qquad + R(T-t)\left[ a(t)+ br_t + \sum_{j=1}^N c_j r_{t-\tau_j} \right] dt  + \sigma(t) R(T-t) dW_t \nonumber\\
    & \qquad + \sum_{j=1}^N c_j \left[ R(T-t-\tau_j)  r_t  - R(T-t) r_{t-\tau_j} \right]dt \label{eq:DerftT}.
\end{align}  
The claimed result comes from the initial value problem in equation \eqref{eq:Rdeter}.
\end{proof}

Observe from Proposition \ref{prop:Forward} that the instantaneous forward rate can be written as
 \[
 df(t,T) =  \sigma(t,T) \left(\int_t^T \sigma(t,s) ds\right)  dt  + \sigma(t,T)dW_t
 \] where $\sigma(t,T) =\sigma(t) R(T-t)$. To see this, notice from \eqref{eq:derD} that \begin{equation*}
     -\sigma(t) D(T-t)  = \sigma(t) \int_0^{T-t} R(s)ds =  \int_t^T \sigma(t) R(s-t) ds.
 \end{equation*}
 This implies that the forward rate satisfies the risk-neutral dynamics of the Heath-Jarrow-Morton model.

\section{Implied short rates}\label{sect:yield}
The proposed model depends on the function $\phi$, which represents the initial value of the short rate on the interval $(-\infty,0]$. Since the short rate value is not observed, one can try to obtain the short rate from the price of zero-coupon bonds given by the market. As we have seen in Section \ref{Sect:ZCB}, we can price zero-coupon bonds, and thus, we can recover the function $\phi$ using the yield curve.  In this section, we assume that the model proposed in equation \eqref{eq:r} has only one delay parameter.

For $t\in[0,\tau_1]$, the dynamics \eqref{eq:r} for the short rate can be written as
\[dr_t = \left( \nu(t)  + b r_t\right)dt + \sigma dW_t,\]
where the function $\nu:[0,\tau_1]\rightarrow\mathds{R}$ is defined as
\begin{equation}
 \nu(s) = a(s) + c_1 \phi(s-\tau_1)  \text{ for } s\in[0,\tau_1] \label{eq:nudef}.
\end{equation}
This implies that the short rate behaves like the Hull-White model \citep{hull1994numerical} when $t\le \tau_1$. Hence we can replicate the yield curve perfectly for maturities $T\in[0,\tau_1]$ and compute the implied short rate on the interval $[-\tau_1,0]$. To obtain perfect replication of the yield curve for maturities smaller or equal to $\tau_1$, we need to pick $\nu$ as
\begin{align}\label{eq:nu}
\nu(s) = \frac{\partial f^M(0,s)}{\partial s} - b f^M(0,s) - \frac{\sigma^2}{2b} \left( 1- e^{2bs}\right) \text{ for } s\in [0,\tau_1], 
\end{align} where $\left(f^M(0,s)\right)_{s\in[0,\tau_1]}$ represents the instantaneous market forward rates at time $0$ \citep[p.~73]{brigo2006interest}. In addition to this, we also need to impose the initial condition \begin{align} \label{eq:r0Forward}
    r_0 & = f^M(0,0). 
\end{align}
We can now use this to recover $\phi$. From \eqref{eq:nudef}, we get that
\begin{equation}\label{eq:phidef}
    \phi(s) = \frac{1}{c_1} (\nu(s+\tau) - a(s)) \text{ for } s \in [-\tau_1,0],
\end{equation}
where $\nu$ is given by \eqref{eq:nu}. Observe that $r_0=\phi(0)$ and hence from \eqref{eq:nudef} and \eqref{eq:r0Forward}, we impose that
\[ f^M(0,0) =  \frac{1}{c_1}(\nu(\tau_1) -a(0)).\] This means that $r_0=f^M(0,0)$. 

Hence, if the function $\phi$ satisfies equation \eqref{eq:phidef}, the model can perfectly replicate the prices of bonds with maturity less or equal to $\tau_1$. The parameters $a$, $b$, $c_1$, and $\sigma$ can be calibrated using zero-coupon bonds with maturities bigger than $\tau_1$. In this we can recover the implied function $\phi$ which is obtained  by the market. That is, this implied function $\phi$ in \eqref{eq:phidef} represents the implied short rates given by the market. 

\section{Caplets}\label{sec:caplets}

We now consider pricing caplets on backward-looking rates. In this section, we build on the ideas presented by Lyashenko and Mercurio \cite{lyashenko10looking} and adapt them to the specific characteristics of our model. Define the money market account $M=(M(t))_{t\geq 0}$ as
\[ M(t) = e^{\int_0^tr_udu} \text{ for all } t\ge0,\]
and then the extended zero coupon bond $\left(B^*(t,T)\right)_{t\ge0}$, as
\begin{align}\label{eq:ZCBExtended}
    B^*(t,T)  &= \begin{cases}
      E_{\mathds{Q}}\left[\left. e^{\int_t^T r_u du} \right| \mathcal{F}_t \right]   & \text{ if } t\in [0,T] \\
         e ^{\int_T^t r_u du} & \text{ if }  t\geq T,
    \end{cases}  \nonumber \\
    &= \begin{cases}
      B(t,T)  & \text{ if } t\in [0,T] \\
         \frac{M(t)}{M(T)} & \text{ if }  t\geq T.
    \end{cases}  
 \end{align}
From equation \eqref{eq:BondRNP}, the extended zero coupon bond can be written as\begin{equation}\label{eq:zcbSol}
     B^*(t,T) = B(0,T) \exp\left\{ \int_{0}^t \left( r_s - \frac{\sigma^2(s)}{2} D^2(T-s) \right) ds + \int_0^t \sigma(s)D(T-s) dW_s \right\}
 \end{equation} for  $t\in[0,T]$. Because $D(t) =0$ for $t\in (-\infty,0]$, the definition of extended zero coupon bonds in \eqref{eq:ZCBExtended} still satisfies \eqref{eq:zcbSol} when $t> T$. To see this, notice that for $t>T$, we have
 \begin{align*}
& B(0,T) \exp\left\{ \int_{0}^t \left( r_s - \frac{\sigma^2(s)}{2} D^2(T-s) \right) ds + \int_0^t \sigma(s) D(T-s) dW_s \right\}\\
     & =  B(0,T) \exp\left\{ \int_{0}^T \left( r_s - \frac{\sigma^2(s)}{2} D^2(T-s) \right) ds + \int_0^T \sigma(s) D(T-s) dW_s \right\}\\
     &\quad \cdot \exp\left\{ \int_{T}^t \left( r_s - \frac{\sigma^2(s)}{2} D^2(T-s) \right) ds + \int_T^t \sigma(s) D(T-s) dW_s \right\}\\
     & = B(T,T)  e^{ \int_T^t r_s ds} = \frac{M(t)}{M(T)} = B^*(t,T).
 \end{align*}It follows that the extended zero coupon bond also satisfies the stochastic differential equation \eqref{eq:BondRNP} for all $t\ge0$. This allows us to define a new measure using the extended zero coupon bond as numeraire, called the extended $T$-forward measure \citep[Section 4.2.4]{andersen2010interest}. This new measure combines the usual $T$-forward measure up to time $T$ with the risk-neutral probability $\mathds{Q}$ after time $T$. To construct the extended $T$-forward measure, we define the process $Z = (Z_t)_{t\geq 0}$ as 
 \[ Z_t = \exp\left\{ - \frac{1}{2} \int_0^t \sigma^2(s) D^2(T-s) ds + \int_0^t \sigma(s) D(T-s) dW_s\right\} \text{ for } t\geq 0.\] 
 Since $D$ is a deterministic continuous function, the process $Z$ is a martingale with respect to $(\mathcal{F}_t)$. Note that $Z_t = Z_T$ for all $t\ge T$. By Girsanov Theorem,  there is an equivalent measure $\mathds{Q}^{T^*}$ with Radon-Nikodym density
 \begin{equation}\label{eq:Q*}
     \left.\frac{d \mathds{Q}^{T^*}}{d\mathds{Q}}\right|_{\mathcal{F}_t}  = Z_t = Z_{T\wedge t}\text{ for } t\geq 0
 \end{equation} with respect to $\mathds{Q}$ and the process $W^{T^*} = \left(W_t^{T^*}\right)_{t\geq 0}$ defined as \[ W_t^{T^*} = W_t - \int_0^t \sigma(s) D(T-s) ds \text{ for } t\geq 0\] is a Brownian motion under $\mathds{Q}^{T^*}$. The measure $\mathds{Q}^{T^*}$ in \eqref{eq:Q*} behaves like the usual $T$-forward measure when $t\leq T$. 
 
 In the case $t> T$, the measure $\mathds{Q}^{T^*}$ acts like the risk-neutral measure. Furthermore, by application of Bayes formula \citep[Theorem A.113]{pascucci2011pde}, we have for every $H\geq t$ that
 \begin{align}
     E_{T^*}\left[\left.f(X) \right| \mathcal{F}_t \right]  & = \frac{ E_{\mathds{Q}}\left[\left.f(X) Z_{T\wedge H}\right| \mathcal{F}_t \right]}{E_{\mathds{Q}}\left[\left. Z_{T\wedge H}\right| \mathcal{F}_t \right]} \nonumber\\
     & = \frac{ Z_{T\wedge H} E_{\mathds{Q}}\left[\left.f(X) \right| \mathcal{F}_t \right]}{Z_{T\wedge H}} = E_{\mathds{Q}}\left[\left.f(X) \right| \mathcal{F}_t \right] \text{ for } t>T, \label{eq:QequalQ*}
 \end{align} where $E_{T^*}\left[\left. . \right| \mathcal{F}_t \right]$ is the conditional expectation with respect to $\mathds{Q}^{T^*}$, $f$ is a bounded measurable function and $X$ is a random variable. Equation \eqref{eq:QequalQ*} tells us that the conditional distribution of any random variable $X$ given $\mathcal{F}_t$ is the same under the risk-neutral measure $\mathds{Q}$ and under the equivalent measure $\mathds{Q}^{T^*}$ when $t>T$. Hence the measure $\mathds{Q}^{T^*}$ defined in \eqref{eq:Q*} is an extended $T$-forward measure.

\subsection{Risk-free rates}

Risk free rates like SONIA, SOFR, and €STR are published on a daily basis. The actual daily-compounded setting-in-arrears rate on the time period $[S,T)$ is given by
\begin{equation}\label{eq:dailyrates}
    \frac{1}{\Delta} \left[ \prod_{i=1}^n (1+r_i \delta)-1\right],
\end{equation} where $r_i$ is the risk free rate on day $i$, $\delta$ is the day count fraction, $n$ is the number of business days in the time interval $[S,T)$ and $\Delta = T-S$. From \eqref{eq:dailyrates},  one can obtain the continuous time approximation
\begin{align}
    R(S,T) & = \frac{1}{\Delta} \left[ e^{\int_S^T r_u du} - 1\right]  = \frac{1}{\Delta} \left[ \frac{M(T)}{M(S)} - 1\right]. \label{eq:RFR} 
\end{align} This continuous time approximation is obtained by taking the limit in \eqref{eq:dailyrates} when $\delta \to 0$. From the continuous time approximation in \eqref{eq:RFR}, is it possible to define two types of forward rates, namely the backward looking forward rate, which is defined as
    \[ R_{S,T}(t) = E_{T^*}\left[ \left. R(S,T)\right| \mathcal{F}_t\right] \text{ for } t\geq 0,\]
and the forward looking rate, which is defined as
    \[F(S,T) = E_{T^*}\left[ \left. R(S,T)\right| \mathcal{F}_S\right].\]
     Notice that the forward rate is a particular case of the backward looking rate with $t=S$. It can be proved that the backward looking rate can be expressed in terms of extended zero coupon bonds, namely
    \begin{align}
        R_{S,T}(t) & = \begin{cases}
       \frac{1}{\Delta} \left[\frac{B(t,S)}{B(t,T)} -1\right]      &  \text{ when } t \leq S, \\
         \frac{1}{\Delta} \left[\frac{M(t)/M(S)}{B(t,T)} -1\right]    & \text{ when } S < t < T,\\
         \frac{1}{\Delta} \left[\frac{M(T)}{M(S)} -1\right]    & \text{ when } t \geq  T
        \end{cases} \nonumber\\
        & = \frac{1}{\Delta} \left[ \frac{B^*(t,S)}{B^*(t,T)} - 1\right]  \text{ for } t\geq 0 \label{eq:backward_rate}
    \end{align} \citep[p. 12]{cao2023lmm}. We will prove that it is possible to obtain an analytical formula for caplets on the backward looking rate \eqref{eq:backward_rate}.

\subsection{Dynamics of the backward looking rate}

We are interested in studying the dynamics of the backward looking rate  $R_{S,T}$ in the extended $T$-forward measure. We derive the dynamics of the backward-looking rate from the dynamics of the extended zero-coupon. The backward looking rate can be expressed as
\begin{equation}\label{eq:RY}
    R_{S,T}(t) = \frac{1}{\Delta} \left[ Y_{S,T}(t) - 1\right],
\end{equation} where $Y_{S,T}(t) = \frac{B^*(t,S)}{B^*(t,T)}$. We will show that is easier to model the bond ratio $Y_{S,T}$ than the backward rate $R_{S,T}$. Equation \eqref{eq:zcbSol} and the It\^o formula give us the stochastic differential equation
\begin{align}
    dY_{S,T}(t) & = Y_{S,T}(t) \sigma(t)\left(D(S-t) - D(T-t) \right)\left[dW_t - \sigma(t) D(T-t) dt\right] \nonumber\\
    & = Y_{S,T}(t) \sigma(t)\left(D(S-t) - D(T-t) \right)dW_t^{T^*}.\label{eq:YSDE}
\end{align} It can be observed that under the extended $T$-forward measure, the process $Y_{S,T}$ is a martingale and satisfies a log-normal distribution. Equations \eqref{eq:RY}--\eqref{eq:YSDE} and the It\^o formula allow us to express the backward looking rate as
\begin{equation}
     dR_{S,T}(t)  =  \left( R_{S,T}(t) + \frac{1}{\Delta}\right) \sigma(t) \left( D(S -t) - D(T-t) \right)dW_t^{T^*}. \label{eq:SDER}
\end{equation}  Notice that the backward looking rate is also a martingale under the extended $T$-forward measure and does not satisfy a normal nor log-normal distribution. For this reason, working with $Y_{S,T}$ is more convenient than with $R_{S,T}$.

\begin{proposition}
The backward looking interest rate $R_{S,T}$ satisfies
\begin{align}
R_{S,T}(t) &=  R_{S,T}(0) \psi_t - \psi_t \int_0^t \frac{1}{ \Delta \psi_u } \sigma^2(u) \left( D(S - u) - D(T-u) \right)^2du \nonumber\\
& \qquad + \psi_t \int_0^t \frac{1}{\Delta \psi_u}  \sigma(u) \left( D(S -u) - D(T-u) \right)dW_u^{T^*}, \label{eq:Rsol}
\end{align} where the process $\psi = \left(\psi_t\right)_{t\geq 0}$ is defined as
\begin{align}
     \psi_t  & = \exp\left\{ -\frac{1}{2}\int_0^t \sigma^2(u) \left( D(S -u) - D(T-u) \right)^2 du \right\} \nonumber \\
      & \cdot \exp\left\{\int_0^t \sigma(u) \left( D(S -u) - D(T-s) \right) dW_u^{T^*} \right\} \text{ for } t \geq 0. \label{eq:psi}
\end{align}
\end{proposition}

\begin{proof}
To solve the stochastic differential equation \eqref{eq:SDER}, we proceed as in the work of Kloeden and Platen \cite[Section 4.2]{platen1992Numerical}. Observe that the process $\psi$ in equation \eqref{eq:psi} satisfies the stochastic differential equation
$$d\psi_t = \psi_t \sigma(t) \left( D(S-t) - D(T-t) \right) dW_t^{T^*}.$$
By the It\^o formula, we have
\begin{multline*}
    d\left(\frac{1}{\psi_t} \right) =\left( \frac{1}{\psi_t}  \sigma^2(t) \left( D(S -t) - D(T-t) \right)^2 \right)dt \\
     -\frac{1}{\psi_t}  \sigma(t) \left( D(S -t) - D(T-t) \right) dW_t^{T^*}.
\end{multline*}
By the It\^o product rule, we have that
\begin{multline}
  d\left(\frac{R_{S,T}(t)}{\psi_t}\right) 
   = - \frac{1}{\Delta \psi_t} \sigma^2(t) \left( D(S -t) - D(T-t) \right)^2dt \\
   + \frac{1}{\Delta \psi_t}  \sigma(t) \left( D(S -t) - D(T-t) \right)dW_t^{T^*}. \label{eq:RSDE2}
\end{multline} From equation \eqref{eq:RSDE2}, one can trivially obtain the result in \eqref{eq:Rsol}.
\end{proof}

\begin{remark}
    The volatility of the backward looking rate $R_{S,T}$ in \eqref{eq:SDER} is given by the function
    \[ g(u) = \sigma(u) \left( D(S -u) - D(T-u) \right) \text{ for } u \geq 0.\] Since $D(u) = 0$ for $u\in (-\infty,0]$, $g(u) =0 $ for $u\geq T$, so the volatility disappears for values of $t$ that are bigger than $T$. From equation \eqref{eq:derD}, the derivative of $g$ satisfies\[
         g'(u) = \sigma'(u) \left(D(S-u) - D(T-u)\right)
          - \sigma(u) \left(R(S-u) - R(T-u)\right)
    \] for all $u\ge0$. If we assume that $\sigma(t) = \sigma >0$ for all $t\geq 0$, then using the definition of the function $R$ in \eqref{eq:R}, we arrive at
    \begin{align*}
        g'(u) &  = - \sigma R(T-u)\\
        & = - \sigma \sum_{n=0}^{\lfloor\frac{T-u}{\tau_1}\rfloor} \sum_{|\alpha|= n} \frac{c^\alpha}{\alpha !} \left(T-u- \langle \alpha,\tau \rangle\right)^n e^{b\left(T-u-\langle \alpha,\tau \rangle\right)} H\left(T-u- \langle \alpha,\tau \rangle\right)
    \end{align*} for $u\in [S,T]$. Notice that if $c^\alpha\geq0$, the derivative is negative on the interval $[S,T]$. We have proved that the volatility decays on the interval $[S,T]$ and disappears on the interval $(T,\infty)$. This is consistent with the findings of  Lyashenko and Mercurio \cite{lyashenko10looking}.
\end{remark}

\subsection{Caplet formula}
Since there are two types of interest rates, it is possible to define two types of caplet payoffs with strike $K$ and expiration date $T$, namely the forward looking caplet with payoff $\Delta\left( R_{S,T}(S) - K\right)^+ $ and the backward looking caplet with payoff $\Delta\left( R_{S,T}(T) - K\right)^+ $.
The price of a caplet under the extended $T$-forward measure is
\begin{align}\label{eq:Caplet}
    V_{\text{Caplet}}(t,\ell) & = \Delta B(t,T) E_{T^*}\left[ \left. (R_{S,T}(\ell) -K )^+ \right| \mathcal{F}_t\right] 
\end{align} for $\ell\in [S,T]$ and $t\geq 0$ with $\ell \geq t$.  If $\ell = S$, equation \eqref{eq:Caplet}, gives the price of a forward looking caplet, while if $\ell = T$, it gives the price of a backward looking caplet. From equations \eqref{eq:RY} and \eqref{eq:Caplet},  we have that
\begin{align}
    V_{\text{Caplet}}(t,\ell)  & = \Delta B(t,T) E_{T^*}\left[ \left. (R_{S,T}(\ell) -K )^+ \right| \mathcal{F}_t\right] \nonumber\\
    & = \Delta B(t,T) E_{T^*}\left[ \left.  \left(\frac{1}{\Delta} (Y_{S,T}(\ell)-1) -K \right)^+ \right| \mathcal{F}_t\right] \nonumber\\
    & =  B(t,T) E_{T^*}\left[ \left.  \left(Y_{S,T}(\ell)- 1 -K \Delta \right)^+ \right| \mathcal{F}_t\right] \nonumber\\
    & =  B(t,T) E_{T^*}\left[ \left.  \left(Y_{S,T}(\ell)- \hat{K} \right)^+ \right| \mathcal{F}_t\right], \label{eq:Caplet2}
\end{align} where $\hat{K} = 1+ K\Delta$. 

The next result gives the analytical formula for the price of a forward looking caplet and a backward looking caplet.
\begin{proposition}
The price of the caplet in \eqref{eq:Caplet} at time $t\geq 0$, with $\ell\geq t$ and $\ell\in [S,T]$, is given by
\begin{equation}\label{eq:CapletFormula}
    V_{\text{Caplet}}(t,\ell) = B(t,T) \left[ Y_{S,T}(t) N(d_+) - \hat{K} N(d_-)  \right],
\end{equation}where
\begin{align*}
d_{\pm} & = \frac{\ln \frac{Y_{S,T}(t) }{\hat{K}}  \pm \frac{1}{2} \nu(t,\ell)}{\sqrt{\nu(t,\ell)}}\\    
\nu(t,\ell) & = \int_t^{\ell} \sigma^2(u) \left( D(S -u) - D(T-u) \right)^2 du.
\end{align*}\end{proposition}
\begin{proof}  Formula \eqref{eq:CapletFormula} follows from the fact that
$$ \left. \ln Y_{S,T}(\ell)\right| \mathcal{F}_t \sim N \left( \ln Y_{S,T}(t) - \frac{1}{2}v(t,\ell), v(t,\ell) \right),$$ equation \eqref{eq:Caplet2}, and the generalized Black-Scholes-Merton call option formula \citep[p.~148]{wilmott2006paul}.
\end{proof}
Equation \eqref{eq:CapletFormula} with $\ell =S$ gives us the price of the forward looking caplet while it gives us the price of the backward looking caplet when $\ell =T$.

\section{Numerical experiments and practical applications}\label{sec:num_exp_pract_app}

We implement the results shown in Sections \ref{Sect:ZCB}--\ref{sec:caplets}. First, we estimate the parameters against time series data and show that the proposed model in \eqref{eq:r} offers an improvement against the model without delay parameters; see Section \ref{sect:estimation}. In Section \ref{sect:pract_app_yield}, we use the yield curve data extracted from the United States Federal Reserve to compute the implied short rates. Additionally in Section \ref{sect:pract_app_caplet}, we use caplet data from LSEG Data \& Analytics for further calibration of the proposed model. For the sake of simplicity, we use a single delay parameter in the experiments shown in Sections  \ref{sect:pract_app_yield}--\ref{sect:pract_app_caplet} and the functions $a$ and $\sigma$ are assumed to be constants in all the experiments.

\subsection{Estimation} \label{sect:estimation}

In this section, we would like to estimate the parameters of the model \eqref{eq:r} using time series data. However, the short rate is not directly observed; we use the daily United States six-month treasury rate as a proxy for the short rate. The data is taking from \href{https://fred.stlouisfed.org/series/DGS6MO}{United States Federal Reserve} , from January~01,~2015 to January~01,~2025. There are $M\in \mathds{N}$ observations equally spaced in time, with a time step $\Delta= \frac{1}{252}$. We also define the time $t_i=i\Delta$ for $i=1,\ldots,M$. Using the It\^o formula, it is not difficult to see that the short rate at time $t_{i}$ can be written as 
\begin{align*}
    r_{t_i} & = r_{t_{i-1}}e^{b\Delta}  + a \int_{t_{i-1}}^{t_i} e^{-b(s-t_{i})} ds + \sum_{j=1}^N  c_j \int_{t_{i-1}}^{t_i} e^{-b(s-t_i)} r_{s-\tau_j} ds + \sigma \int_{t_{i-1}}^{t_i} e^{-b(s-t_i)} dW_s.
\end{align*} If the time step $\Delta \leq \tau_1$, we have that
\[ r_{t_i} \left|\mathcal{F}_{t_{i-1}} \right. \sim N\left( m_{t_i},s^2_{t_{i}} \right)\] with
\begin{align}
    m_{t_i} & =   r_{t_{i-1}}e^{b\Delta}  + a \int_{t_{i-1}}^{t_i} e^{-b(s-t_{i})} ds + \sum_{j=1}^N  c_j \int_{t_{i-1}}^{t_i} e^{-b(s-t_i)} r_{s-\tau_j} ds, \label{eq:integralm}\\
    s^2_{t_{i}} & = \sigma^2 \int_{t_{i-1}}^{t_i} e^{-2 b(s-t_i)} ds.\label{eq:integrals}
\end{align} The Riemann integrals that appear in \eqref{eq:integralm}--\eqref{eq:integrals} can be approximated using the trapezoid method (see \cite{Cretarola_FigaTalamanca_Patacca2020}), hence the parameters $a,b,c_1,c_2,\ldots,c_N$ and $\sigma$ can be estimated using linear regression.  The delay parameters $\tau_1,\tau_2,\ldots,\tau_n$ are selected using the periodogram. 

 First, we use the periodogram to estimate the delay parameters $\tau_1,\tau_2,\ldots,\tau_n$. The periodogram allows us to identify the main Fourier frequencies of the time series increments. Peaks of the periodogram correspond to dominant cyclical components of the
series. To that end, we use the ``periodogram" function included in the Python library SciPy \citep{2020SciPy-NMeth}. Once the periodogram is computed, we select the frequencies that are above the  $95\%$ confidence band for white noise \citep[p. 539]{percival2020spectral}. We rank those frequencies in descending order by the periodogram's value. For each selected frequency $f_j$, we compute the value of the delay as $\tau_j = \nicefrac{1}{f_j}$ .  Once we have computed the delay parameters, we fit the model by adding them one at a time.

 Initially, we fit the model with no delays; then, we fit the model with just one delay, whose frequency has the highest peak in the periodogram. After that, we fit the model with two delay parameters whose frequencies have the two highest peaks in the periodogram. We continue in this way until all selected delays are included in the model. The Akaike information criterion (AIC) \citep[p.~605]{deleeuw1992introduction} is shown in Figure \ref{fig:estimation_mse}. According to the AIC, we should select the model with $8$ delay parameters. The estimated parameters for the model with the $8$  delay parameters are $\hat{a} = 0.3211$, $\hat{b} = 6.3050$, $\hat{\sigma} = 0.4664$, and the estimated delay parameters, and their coefficients are given as follows:
  
\[
\begin{array}{c|cccccccc}
j          & 1        & 2        & 3        & 4        & 5        & 6        & 7        & 8 \\ \hline
\hat\tau_j & 0.0104   & 0.0147   & 0.0194   & 0.0254   & 0.0322   & 0.0402   & 0.0781   & 0.2540 \\
\hat c_j   & -1.923   & -4.942   & -16.775  & 0.667    & -12.548  & 1.867    & 31.865   & -14.554.
\end{array}
\] For the estimated model with eight delay parameters, we can not guarantee that the model has a limiting distribution; see Proposition \ref{prop:stationary_distr}. Furthermore, in Figure \ref{fig:estimation_mse}, we also show the value of the delay parameter that is added in each estimation step. The first delay value included in the model is $0.0194$. Notice that several delay values are smaller than $\nicefrac{1}{12}$, implying that the proposed model captures short-term memory behavior. However, longer memory behavior is also captured, as evidenced by a delay of $0.254$.

\begin{figure}[h!]
\centering
\begin{tikzpicture}
  \begin{axis}[
      height=6cm ,
       width=10cm, 
      axis y line*=left,
      axis x line*=bottom,
      xlabel={Number of delays},
      ylabel={AIC},
      xmin=0, xmax=8,
      grid=both,
      legend style={at={(0.02,0.98)},anchor=north west,draw=none,fill=none},
  ]
    \addplot+[blue,mark=*] table[
        x=index,
        y=Akaike,
        col sep=comma
    ]{DataForGraphs.csv};
    \addlegendentry{Akaike information}
  \end{axis}

   \begin{axis}[
       height=6cm ,
   width=10cm,
      axis y line*=right,
      axis x line=none,
      ylabel={Delays},
      xmin=0, xmax=8,
      legend style={at={(0.98,0.98)},anchor=north east,draw=none,fill=none}
  ]
  \addplot+[green,mark=square*, mark options={fill=green, draw=green}] table[
        x=index,
        y=Delays,
        col sep=comma,
    ]{DataForGraphs.csv};
    \addlegendentry{Delays}
  \end{axis}
\end{tikzpicture}

 \caption{Akaike information criterion obtained by the fitted model with a different number of delay parameters. Delay parameter that is added at each estimation step.}
 \label{fig:estimation_mse}
\end{figure}

 Finally, we study the presence of autocorrelation in the residuals. To that end, we apply the Ljung–Box test with $1$ lag to the residuals obtained by the model with different numbers of delays. The p-values obtained by the fitted models are shown in Figure \ref{fig:estimation_Ljung_Box}. The p-value obtained by the model with no delays is $0.0063$. Hence, we should reject the null hypothesis that the residuals have no correlation. However, the inclusion of the delay parameters into the model gives us a p-value that does not allow us to reject the null hypothesis. This shows that the introduction of delay parameters captures short-memory behavior that the model with no delay parameters is not able to capture.  

\begin{figure}[h!]
\centering
\begin{tikzpicture}
  \begin{axis}[
    xlabel={Number of delays},
    ylabel={P-value},
    grid=major,          
    width=10cm,          
    height=6cm           
  ]
    \addplot table [x=index, y=Lag1, col sep=comma] {DataForGraphs.csv};

  \end{axis}
\end{tikzpicture}
 \caption{P-value of the Ljung–Box test with 1 lag obtained by the fitted model's residuals with different delay parameters.}
 \label{fig:estimation_Ljung_Box}
\end{figure}

\subsection{Implied short rates}\label{sect:pract_app_yield}

In this section, we use the method explained in Section \ref{sect:yield} to obtain the implied short rates, to that end we use the spot yield curve on April~19,~2024. The spot yield curve represents the yield to maturity for zero coupon bonds. To construct this curve, we use the par yield curve provided by the \href{https://home.treasury.gov/resource-center/data-chart-center/interest-rates/TextView?type=daily_treasury_yield_curve&field_tdr_date_value_month=202409}{United States Federal Reserve} for maturities up to one year (as Treasury bills are effectively zero coupon bonds at these maturities). We extend the spot yield curve for longer maturities (annually, up to 15 years) using the estimates provided by the \href{https://www.federalreserve.gov/data/nominal-yield-curve.htm}{United States Federal Reserve}. These estimates are constructed using the Nelson-Siegel-Svensson model \citep{Svensson1994}; see also the work of Refet, Sack and Wright \cite{Refet2007}. The yield curve is then used to determine market zero coupon bond prices. The United States spot yield curve is depicted in Figure \ref{fig:yieldcurve}. Notice that the yield curve is inverted, meaning that the yield on short-term bonds is higher than that of long-term bonds.

\begin{figure}[h!]
\centering\begin{tikzpicture} 
 \begin{axis}[
 xlabel={Maturity (years)},
width=11cm,height=8cm,
legend pos=north east,
legend style={fill=white,fill opacity=0.6,text opacity=1},
SimStyle]
\addplot[
 mark=square,
 color=blue,
]
 coordinates {
(0.083, 0.0549) (0.167, 0.0550999999999999) (0.25, 0.0545) (0.333, 0.0544) (0.5, 0.0538999999999999) (1.0, 0.0516999999999999) (2.0, 0.049088) (3.0, 0.0476689999999999) (4.0, 0.046694) (5.0, 0.046067) (6.0, 0.045709) (7.0, 0.045555) (8.0, 0.045552) (9.0, 0.0456569999999999) (10.0, 0.045837) (11.0, 0.046063) (12.0, 0.0463159999999999) (13.0, 0.046578) (14.0, 0.046836) (15.0, 0.047081)
};
\end{axis} 
 \end{tikzpicture}
 \caption{United States spot yield curve on April 19, 2024.}
 \label{fig:yieldcurve}
\end{figure}

We obtain the implied function $\phi$ using zero coupon bonds with maturities up to $\tau_1$, while zero coupon bonds with maturities greater than $\tau_1$ are used to calibrate the parameters $a$, $b$, $c_1$, and $\sigma$. In our experiment, we set the delay parameter $\tau_1$ to four values: $1$ year, $2$ years, $3$ years, and $4$ years. This means we obtain the implied short rates for each delay parameter, with $\tau_1 = 1$, $\tau_1 = 2$, $\tau_1 = 3$, and $\tau_1 = 4$. For each $\tau_1$, the implied function $\phi$ is expected to capture information about the short-rate dynamics over the preceding $\tau_1$ years. The initial value of the implied function $\phi$ is derived from \eqref{eq:phidef}, and the parameters $a$, $b$, $c_1$, and $\sigma$ are calibrated by minimizing the mean squared error, i.e.
\begin{equation}\label{eq:MSEbonds}
    \text{MSE}_{\text{Bonds}} = \frac{1}{N_{\tau_1}}\sum_{j\in I_{\tau_1}} \left( B_{\text{Market}}(0,j) - B_{\text{Model}}(0,j)\right)^2,
 \end{equation} Here, \( B_{\text{Market}} \) represents the market price of a zero-coupon bond, while \( B_{\text{Model}} \) denotes the zero-coupon bond price as predicted by the model. \( N_{\tau} \) is the number of bonds with maturities greater than \( \tau_1 \), and \( I_{\tau_1} \) is the set of maturities greater than \( \tau_1 \). The implied \( \phi \) for each \( \tau_1 \) is presented in Figure \ref{fig:impliedphi}. Notice that, for each \( \tau_1 \), the implied \( \phi \) is inconsistent with the model, as \( \phi \) should be independent of the chosen delay parameter \( \tau_1 \). Nevertheless, the implied \( \phi \) suggests that interest rates rose, stabilized, and then slightly declined in the years preceding April 19, 2024. This pattern aligns roughly with historical trends in short-term rates, as indicated by the United States three-month Treasury bill yield data from the \href{https://fred.stlouisfed.org/series/DTB3}{FRED database}. However, the implied \( \phi \) does not closely match the actual interest rate values over this period. One potential improvement for the model could be the use of a random initial function \( \phi \) rather than a deterministic one.

The calibrated parameter values for each \( \tau_1 \) are presented in Table \ref{table:param_cali_partial}. Notice that all parameters, except for \( c_1 \), remain consistent across different values of \( \tau_1 \). For each \( \tau_1 \), these parameter values result in a model with a limiting distribution, as established in Proposition \ref{prop:stationary_distr}.

\begin{table}
\tbl{Parameter values obtained by partial calibration.}
{\begin{tabular}{|c |S[table-format=2.5] |S[table-format=2.5] |S[table-format=2.5] |S[table-format=2.5] | }
\hline
Parameter & \multicolumn{1}{c|}{$\tau_1 =1$} & \multicolumn{1}{c|}{$\tau_1 =2$} & \multicolumn{1}{c|}{$\tau_1 =3$} & \multicolumn{1}{c|}{$\tau_1 =4$}  \\\hline
\hline
$a$ & 0.05219 &  0.06048 & 0.06382 & 0.06383 \\
\hline
$b$ & -1.00232  & -1.00156 & -1.00027  & -0.99856 \\
\hline
$c_1$ & -0.14587 & -0.31822  & -0.38127 & -0.37103\\ 
\hline
$\sigma$ & 0.00402 & 0.00534 & 0.00638 & 0.00799 \\
\hline
\end{tabular}}
\label{table:param_cali_partial}
\end{table}
In Figure \ref{fig:cal_yield_partial}, we show the market zero coupon bond price and the bond price given by the model. Figure \ref{fig:cal_yield_partial} additionally contains the absolute error between the market price and the model price, defined as
 \[\text{Absolute Error } = \left| B_{\text{Market}}(0,T) - B_{\text{Model}}(0,T) \right|.\]Note that errors for maturities smaller than $\tau_1$ are negligible compared to the maturities bigger than $\tau_1$.

\begin{figure}[h!]
\centering\begin{tikzpicture} 
 \begin{axis}[
 xlabel={Time (in years)},
width=11cm,height=8cm,
legend pos=south west,
legend style={fill=white,fill opacity=1,text opacity=1},
SimStyle,]
\addplot[
 thick,dashed, color=blue,
]
 coordinates {
(-0.997, 0.0347059070593179) (-0.977, 0.0352910608051327) (-0.956, 0.0358688521253713) (-0.936, 0.0364391843079146) (-0.916, 0.0370022336311642) (-0.895, 0.0375573777701848) (-0.875, 0.0381058290878891) (-0.854, 0.0386467625798516) (-0.834, 0.0391804808499658) (-0.814, 0.0397073844674575) (-0.793, 0.0402269259080873) (-0.773, 0.0407396491560166) (-0.753, 0.0412451984137707) (-0.732, 0.041744135918521) (-0.712, 0.0422364867563774) (-0.692, 0.0427213807668032) (-0.671, 0.0431999391388419) (-0.651, 0.0436718048656758) (-0.631, 0.0441371297586248) (-0.61, 0.0445958549027923) (-0.59, 0.0450480675076475) (-0.57, 0.0454940607652283) (-0.549, 0.0459332212528773) (-0.529, 0.0463662560106602) (-0.509, 0.0467929054127387) (-0.488, 0.0472134297992349) (-0.468, 0.0476277766968663) (-0.448, 0.0480358984512016) (-0.427, 0.048437888072131) (-0.407, 0.0488340811449968) (-0.387, 0.0492240083434335) (-0.366, 0.0496083053833101) (-0.346, 0.0499865693496569) (-0.326, 0.05035910992577) (-0.305, 0.0507254484803847) (-0.285, 0.0510864247836743) (-0.264, 0.0514416406576895) (-0.244, 0.0517910535780615) (-0.224, 0.0521353544572244) (-0.203, 0.05247368127321) (-0.183, 0.052806899697993) (-0.163, 0.053134061455644) (-0.142, 0.0534561855250587) (-0.122, 0.0537730129983335) (-0.102, 0.0540845979762323) (-0.081, 0.0543904380134929) (-0.061, 0.0546913461334654) (-0.041, 0.0549870524849346) (-0.02, 0.0552778508837572) (0.0, 0.0555631803431222)
};
\addplot[
 thick,dotted, color=red,
]
 coordinates {
(-1.997, 0.0420886741644304) (-1.956, 0.0426202531839227) (-1.915, 0.0431381295808836) (-1.875, 0.0436427277075705) (-1.834, 0.0441341629228358) (-1.793, 0.0446126051118473) (-1.752, 0.0450782233256231) (-1.712, 0.0455315081414432) (-1.671, 0.0459721823507033) (-1.63, 0.0464009548825608) (-1.589, 0.046817370752224) (-1.549, 0.0472223124093696) (-1.508, 0.0476153749684292) (-1.467, 0.0479974675099389) (-1.426, 0.0483679969870971) (-1.386, 0.0487275084341761) (-1.345, 0.0490763888197172) (-1.304, 0.0494144207440004) (-1.263, 0.0497418274670441) (-1.223, 0.0500589259213986) (-1.182, 0.0503660966178482) (-1.141, 0.0506631932432532) (-1.1, 0.0509506325487949) (-1.06, 0.0512281695976794) (-1.019, 0.0514962323961611) (-0.978, 0.051755226976647) (-0.937, 0.0520048286801721) (-0.897, 0.0522455945240398) (-0.856, 0.052477444979307) (-0.815, 0.0527006996268635) (-0.774, 0.0529153546706425) (-0.734, 0.0531217522018204) (-0.693, 0.0533198952754247) (-0.652, 0.0535096764662016) (-0.611, 0.0536918107728533) (-0.571, 0.0538659664600579) (-0.53, 0.0540326644878885) (-0.489, 0.0541919750486965) (-0.448, 0.0543436327403629) (-0.408, 0.0544880583910262) (-0.367, 0.054625396771794) (-0.326, 0.0547557476601501) (-0.285, 0.0548792800579426) (-0.245, 0.0549964516535582) (-0.204, 0.0551064982873801) (-0.163, 0.0552104217705321) (-0.122, 0.0553078145153309) (-0.082, 0.0553989608926299) (-0.041, 0.0554841918673816) (0.0, 0.0555631803431222)
};
\addplot[
 thick,densely dashed,  color=black,
]
 coordinates {
(-2.997, 0.0440717328353875) (-2.936, 0.0447297207240684) (-2.875, 0.045362762415358) (-2.813, 0.0459711091797707) (-2.752, 0.0465554498347825) (-2.691, 0.0471165303221906) (-2.63, 0.0476546430291059) (-2.569, 0.0481703520990419) (-2.508, 0.0486641516258569) (-2.446, 0.0491367976888185) (-2.385, 0.0495887209120721) (-2.324, 0.0500203428207792) (-2.263, 0.0504320711282711) (-2.202, 0.0508244001801969) (-2.141, 0.0511980588667516) (-2.079, 0.0515534219674435) (-2.018, 0.0518907661055416) (-1.957, 0.0522105799114836) (-1.896, 0.0525136019913189) (-1.835, 0.0527999716250926) (-1.774, 0.0530704209618777) (-1.713, 0.0533248913972129) (-1.651, 0.053564279337346) (-1.59, 0.053788791086539) (-1.529, 0.0539988256797049) (-1.468, 0.0541946972057045) (-1.407, 0.0543769104125396) (-1.346, 0.0545459203313344) (-1.284, 0.0547019679582822) (-1.223, 0.0548453176897508) (-1.162, 0.0549765264920502) (-1.101, 0.0550958745330944) (-1.04, 0.0552036341504746) (-0.979, 0.0553000897268507) (-0.917, 0.055385747941724) (-0.856, 0.0554609188332429) (-0.795, 0.0555257159179442) (-0.734, 0.0555804195456247) (-0.673, 0.0556256937645761) (-0.612, 0.0556614722320014) (-0.55, 0.0556881363865318) (-0.489, 0.0557059749554374) (-0.428, 0.0557153721189004) (-0.367, 0.0557163358177637) (-0.306, 0.0557094337316167) (-0.245, 0.0556946907228916) (-0.183, 0.0556723885255185) (-0.122, 0.0556430830835619) (-0.061, 0.0556064772745202) (0.0, 0.0555631803431223)
};
\addplot[
  thick, color=green,
]
 coordinates {
(-3.997, 0.045507843341058) (-3.915, 0.0463844885854699) (-3.834, 0.0472172050778273) (-3.752, 0.0480069699328032) (-3.671, 0.0487551928049286) (-3.589, 0.049462781863302) (-3.507, 0.0501313614842691) (-3.426, 0.0507617703482777) (-3.344, 0.0513554695805734) (-3.263, 0.0519133506842829) (-3.181, 0.0524367539565661) (-3.1, 0.0529266672314387) (-3.018, 0.0533843409129354) (-2.936, 0.0538106890829256) (-2.855, 0.0542070466840361) (-2.773, 0.0545740652108054) (-2.692, 0.0549130504116963) (-2.61, 0.0552247031488202) (-2.529, 0.0555103418156762) (-2.447, 0.055770625154603) (-2.366, 0.0560065224877212) (-2.284, 0.056218932799216) (-2.202, 0.0564088070660114) (-2.121, 0.0565768091005949) (-2.039, 0.0567240883345418) (-1.958, 0.0568510498069768) (-1.876, 0.0569585159351364) (-1.795, 0.0570473738970615) (-1.713, 0.0571186257128113) (-1.631, 0.0571723815862251) (-1.55, 0.0572095915269993) (-1.468, 0.0572309801071076) (-1.387, 0.0572373804959219) (-1.305, 0.0572290860506173) (-1.224, 0.0572067118632836) (-1.142, 0.0571713060077106) (-1.06, 0.057122877030673) (-0.979, 0.0570621691404507) (-0.897, 0.05698983282755) (-0.816, 0.0569065232756285) (-0.734, 0.0568123036066569) (-0.653, 0.0567081042881047) (-0.571, 0.056594319333125) (-0.489, 0.0564712029118011) (-0.408, 0.0563392255636465) (-0.326, 0.0561989902573414) (-0.245, 0.0560508737743314) (-0.163, 0.0558952873733146) (-0.082, 0.0557327191583166) (0.0, 0.0555631803431222)
};
\legend{$\tau_1 = 1$,$\tau_1 = 2$,$\tau_1 = 3$,$\tau_1 = 4$ }
\end{axis} 
 \end{tikzpicture}
  \caption{Implied initial value function $\phi$ obtained from \eqref{eq:phidef} obtained for different values of the delay parameter $\tau_1$.}
  \label{fig:impliedphi}
\end{figure}

\begin{figure}
 \begin{center}
\begin{tikzpicture}
\begin{groupplot}[group style={group name=my plots, group size=2 by 1, horizontal sep = 1cm}, width = 7.5cm, height=6cm, SimStyle]

\nextgroupplot[xlabel={Maturity date (years)},legend to name={CommonLegend},legend style={legend columns=5}]

\addplot[
 color=purple,
 mark=o,
]
 coordinates {
(0.083333, 0.9954094937267488) (0.166667, 0.9908964356493172) (0.25, 0.9864576101277268) (0.333333, 0.9820899328776004) (0.5, 0.9735563098872478) (1.0, 0.9493789401949392) (2.0, 0.9058873384763428) (3.0, 0.8666476881525217) (4.0, 0.8299452268424765) (5.0, 0.7947985325289223) (6.0, 0.7606887862350451) (7.0, 0.72738375540808) (8.0, 0.6948217932954353) (9.0, 0.663035561141132) (10.0, 0.6321027160653766) (11.0, 0.6021150591921182) (12.0, 0.573160251428892) (13.0, 0.545311940273575) (14.0, 0.5186253604388571) (15.0, 0.4931363533150869)
};
\addplot[
 color=blue,
 mark=square,
]
 coordinates {
(0.083333, 0.9954036915703092) (0.166667, 0.990885346866608) (0.25, 0.9864417075063324) (0.333333, 0.9820696500021387) (0.5, 0.9735284535939508) (1.0, 0.9493356212816576) (2.0, 0.9055132952388824) (3.0, 0.8649779628377292) (4.0, 0.8264972627182905) (5.0, 0.7897654259831941) (6.0, 0.7546709024888452) (7.0, 0.7211364006992812) (8.0, 0.6890920792801862) (9.0, 0.6584716751484364) (10.0, 0.6292119173562203) (11.0, 0.6012523407941551) (12.0, 0.5745351710200399) (13.0, 0.5490052041450636) (14.0, 0.5246096823284719) (15.0, 0.5012981956203866)
};
\addplot[
 color=red,
  mark=triangle,
]
 coordinates {
(0.083333, 0.9954036965292012) (0.166667, 0.9908853561822332) (0.25, 0.9864417207368568) (0.333333, 0.9820696668950288) (0.5, 0.9735284773307245) (1.0, 0.949335660866582) (2.0, 0.9058319774184664) (3.0, 0.8662794128430703) (4.0, 0.8283324409824883) (5.0, 0.791457503168771) (6.0, 0.7559059041291462) (7.0, 0.7219286492797439) (8.0, 0.6895526708466182) (9.0, 0.658677980448494) (10.0, 0.6291942422700292) (11.0, 0.601021692449088) (12.0, 0.5741032436994258) (13.0, 0.548388478767428) (14.0, 0.5238263901764163) (15.0, 0.5003654729634169)
};
\addplot[
 color=black,
  mark=diamond,
]
 coordinates {
(0.083333, 0.9954037044729692) (0.166667, 0.9908853687571848) (0.25, 0.9864417330727728) (0.333333, 0.982069672816664) (0.5, 0.973528448742905) (1.0, 0.9493353382032328) (2.0, 0.9058305093322112) (3.0, 0.8665879741551739) (4.0, 0.8296435627995373) (5.0, 0.7934581003464573) (6.0, 0.757912091701358) (7.0, 0.7233622101100468) (8.0, 0.6902453888679937) (9.0, 0.6587916894975898) (10.0, 0.6289776579795783) (11.0, 0.6006530670384286) (12.0, 0.5736430976286001) (13.0, 0.5478210998260861) (14.0, 0.5231162376807749) (15.0, 0.4994931409017951)
};
\addplot[
 color=green,
  mark=x,
]
 coordinates {
(0.083333, 0.9954037155742368) (0.166667, 0.9908853899415072) (0.25, 0.9864417639677332) (0.333333, 0.982069713536822) (0.5, 0.9735285110040552) (1.0, 0.94933549323576) (2.0, 0.9058309546658652) (3.0, 0.8665886294804694) (4.0, 0.8298867333787991) (5.0, 0.7945472101514379) (6.0, 0.7596484148462247) (7.0, 0.7252924504375787) (8.0, 0.6918748675548514) (9.0, 0.6597337241752658) (10.0, 0.6291425292325933) (11.0, 0.6001963091729996) (12.0, 0.5728108832667279) (13.0, 0.5468338312452975) (14.0, 0.5221043740194987) (15.0, 0.498484206143804)
};
\legend{Market price, $\tau_1 = 1$, $\tau_1=2$, $\tau_1 = 3$, $\tau_1 = 4$}

\nextgroupplot[ xlabel={Maturity date (years)}]
\addplot[
 color=blue,
 mark=square,
]
 coordinates {
(0.083333, 5.80215643974924e-06) (0.166667, 1.108878270927871e-05) (0.25, 1.5902621394170424e-05) (0.333333, 2.0282875461741234e-05) (0.5, 2.7856293297046086e-05) (1.0, 4.331891328146487e-05) (2.0, 0.0003740432374603) (3.0, 0.0016697253147924) (4.0, 0.003447964124186) (5.0, 0.0050331065457281) (6.0, 0.0060178837461999) (7.0, 0.0062473547087987) (8.0, 0.005729714015249) (9.0, 0.0045638859926956) (10.0, 0.0028907987091563) (11.0, 0.000862718397963) (12.0, 0.0013749195911478) (13.0, 0.0036932638714886) (14.0, 0.0059843218896148) (15.0, 0.0081618423052996)
};
\addplot[
 color=red,
 mark=triangle,
]
 coordinates {
(0.083333, 5.797197547674493e-06) (0.166667, 1.107946708389651e-05) (0.25, 1.5889390869827658e-05) (0.333333, 2.0265982571654465e-05) (0.5, 2.783255652338834e-05) (1.0, 4.327932835723747e-05) (2.0, 5.5361057876290864e-05) (3.0, 0.0003682753094513) (4.0, 0.0016127858599882) (5.0, 0.0033410293601513) (6.0, 0.0047828821058989) (7.0, 0.0054551061283361) (8.0, 0.005269122448817) (9.0, 0.004357580692638) (10.0, 0.0029084737953474) (11.0, 0.0010933667430301) (12.0, 0.0009429922705337) (13.0, 0.003076538493853) (14.0, 0.0052010297375592) (15.0, 0.00722911964833)
};
\addplot[
 color=black,
 mark=diamond,
]
 coordinates {
(0.083333, 5.789253779697745e-06) (0.166667, 1.1066892132371995e-05) (0.25, 1.5877054953983283e-05) (0.333333, 2.0260060936361235e-05) (0.5, 2.7861144342833377e-05) (1.0, 4.360199170627688e-05) (2.0, 5.682914413163154e-05) (3.0, 5.971399734783223e-05) (4.0, 0.0003016640429391) (5.0, 0.001340432182465) (6.0, 0.0027766945336871) (7.0, 0.0040215452980332) (8.0, 0.0045764044274415) (9.0, 0.0042438716435422) (10.0, 0.0031250580857983) (11.0, 0.0014619921536895) (12.0, 0.000482846199708) (13.0, 0.002509159552511) (14.0, 0.0044908772419177) (15.0, 0.0063567875867081)
};
\addplot[
 color=green,
 mark=x,
]
 coordinates {
(0.083333, 5.778152512014856e-06) (0.166667, 1.1045707810075632e-05) (0.25, 1.5846159993615494e-05) (0.333333, 2.0219340778560024e-05) (0.5, 2.779888319259793e-05) (1.0, 4.344695917912755e-05) (2.0, 5.638381047756713e-05) (3.0, 5.905867205224791e-05) (4.0, 5.849346367736708e-05) (5.0, 0.0002513223774843) (6.0, 0.0010403713888204) (7.0, 0.0020913049705012) (8.0, 0.0029469257405838) (9.0, 0.0033018369658662) (10.0, 0.0029601868327833) (11.0, 0.0019187500191185) (12.0, 0.0003493681621641) (13.0, 0.0015218909717225) (14.0, 0.0034790135806416) (15.0, 0.0053478528287171)
};

\end{groupplot}
\node[align=center,anchor=north] at ([yshift=-12mm]my plots c1r1.south) {(a) Zero coupon bond prices.};
\node[align=center,anchor=north] at ([yshift=-12mm]my plots c2r1.south) {(b) Absolute error.};
\path ([yshift=-12mm]my plots c1r1.south east) -- node[below,anchor=north]{\ref{CommonLegend}} ([yshift=-24mm]my plots c2r1.south west);
\end{tikzpicture}
\end{center}
\caption{Zero coupon bond prices and absolute value of the error between market prices and model prices using the method shown in Section \ref{sect:yield}}
\label{fig:cal_yield_partial}
\end{figure}

\subsection{Calibration against caplet prices} \label{sect:pract_app_caplet}

We calibrate the proposed model using US market caplet prices on 1 April 2024 from \href{https://www.lseg.com/en/data-analytics/}{LSEG Data \& Analytics} using \href{https://www.fenicsmd.com/ }{Fenics market data}. The data was downloaded on 4 October 2024. We focus on three-month forward-looking caplets based on SONIA rates, with maturity date from 2 June 2024 to 3 April 2029 and strikes ranging from 3.5\% to 6\%. Expiry dates are expressed in short format indicating the number of years and months; for example, 1Y6M means 1 year and 6 months from 1 April 2024. The data is divided into two sets: a calibration set consisting of caplets with non-zero market prices and short-term maturities from 3M to 4Y (comprising 172 caplets) and an out-of-sample set consisting of long-term maturities from 4Y3M to 5Y (comprising 44 caplets); see Figure \ref{fig:Caplet_price_mkt}. We have selected the calibration set from maturities 3M to 4Y since caplet prices tend to be convex from 3M to 1Y9M, and then they tend to be concave for maturities from 2Y to 4Y, as shown in Figure \ref{fig:mean_caplets}. We selected the calibration set in this way after observing that, on average (see Figure \ref{fig:mean_caplets}), caplet prices tend to be convex for maturities up to around 2 years, and concave after that. This suggests a change point at 2 years, and it would be interesting to verify whether the calibrated parameter $\tau_1$ has a value near 2.

\begin{figure}
 \begin{center}
\begin{tikzpicture}
\begin{groupplot}[group style={group name=my plots, group size=2 by 1, horizontal sep = 1.5cm}, width = 7cm, height=6cm,SimStyle]
\nextgroupplot[xlabel={Strike price}, ylabel={Caplet price},legend to name={CommonLegend},legend style={legend columns=4}]

\addplot[color=blue,mark=square,] 
coordinates {
(0.035, 0.416178894049314) (0.0375, 0.354654014673566) (0.04, 0.293129135297817) (0.0425, 0.231604255922068) (0.045, 0.17007937654632) (0.0475, 0.1085544971705709) (0.05, 0.0470296177948222)
};
\addlegendentry{3M}

\addplot[color=red,mark=triangle,] 
coordinates {
(0.035, 0.350781887263215) (0.0375, 0.289347231019456) (0.04, 0.227923450013974) (0.0425, 0.166765498963497) (0.045, 0.107594920701065) (0.0475, 0.0638369561568008) (0.05, 0.0378190959788312) (0.0525, 0.0238509501261299) (0.055, 0.016090550300391) (0.0575, 0.0111349731323403) (0.06, 0.0080857238217938)
};
\addlegendentry{6M}

\addplot[color=black,mark=diamond,] 
coordinates {
(0.035, 0.280952373032394) (0.0375, 0.220836106712543) (0.04, 0.161517498177964) (0.0425, 0.107447398337205) (0.045, 0.0625098887999757) (0.0475, 0.0424648091454794) (0.05, 0.0310612202057709) (0.0525, 0.0247746598100018) (0.055, 0.0208987325997579) (0.0575, 0.0177480804977739) (0.06, 0.0155472803535003)
};
\addlegendentry{9M}

\addplot[color=green,mark=x,] 
coordinates {
(0.035, 0.2110944547650609) (0.0375, 0.156050728261397) (0.04, 0.104088578137766) (0.0425, 0.0641854255600932) (0.045, 0.0351208487642906) (0.0475, 0.0284145114893137) (0.05, 0.0242933572161671) (0.0525, 0.0222955629836899) (0.055, 0.0211782690739658) (0.0575, 0.0198640161892471) (0.06, 0.0189788028544859)
};
\addlegendentry{1Y}

\addplot[dashed,blue,mark=square,] 
coordinates {
(0.035, 0.166796732512258) (0.0375, 0.122596356134641) (0.04, 0.0843926662124154) (0.0425, 0.0595184607953908) (0.045, 0.0404737326537517) (0.0475, 0.0360789992283108) (0.05, 0.032699712652479) (0.0525, 0.0306443578127743) (0.055, 0.029281268457915) (0.0575, 0.0277676082959408) (0.06, 0.0266519946908044)
};
\addlegendentry{1Y3M}

\addplot[dashed,red,mark=triangle,] 
coordinates {
 (0.035, 0.154378771239264) (0.0375, 0.1204934503925389) (0.04, 0.0921688232767022) (0.0425, 0.0747283221016545) (0.045, 0.0595210133454559) (0.0475, 0.0539754601052054) (0.05, 0.0493825239329591) (0.0525, 0.0458599897933439) (0.055, 0.0431444262440771) (0.0575, 0.0405872458658677) (0.06, 0.0385277796725604)
};
\addlegendentry{1Y6M}

\addplot[dashed,black,mark=diamond,] 
coordinates {
(0.035, 0.150289150930697) (0.0375, 0.126254303220834) (0.04, 0.1065645166434959) (0.0425, 0.0953011372533808) (0.045, 0.0836234132979504) (0.0475, 0.076172044853754) (0.05, 0.0698363897419911) (0.0525, 0.0642605648978354) (0.055, 0.0596183594188027) (0.0575, 0.0555840984901783) (0.06, 0.0521614006908143)
};
\addlegendentry{1Y9M}

\addplot[dashed,green,mark=x,]
coordinates {
 (0.035, 0.150506890644918) (0.0375, 0.135724448464234) (0.04, 0.1240970376234339) (0.0425, 0.1186700061900799) (0.045, 0.110718980469722) (0.0475, 0.100949923583256) (0.05, 0.0925588581615062) (0.0525, 0.0845423277045256) (0.055, 0.0775651743791184) (0.0575, 0.0717364654795286) (0.06, 0.0666283401901068)
 };
\addlegendentry{2Y}

\addplot[dotted,blue,mark=square,] 
coordinates { (0.035, 0.166067066875945) (0.0375, 0.151009128770807) (0.04, 0.13706040430733) (0.0425, 0.129035167305698) (0.045, 0.119204887955828) (0.0475, 0.109305944929208) (0.05, 0.100716136144805) (0.0525, 0.0927864760361426) (0.055, 0.0858278965992008) (0.0575, 0.0798873610890693) (0.06, 0.0746398848479352)
 };
\addlegendentry{2Y3M}

\addplot[dotted,red,mark=triangle,] 
coordinates {
(0.035, 0.185385838342239) (0.0375, 0.168966637272562) (0.04, 0.151538551141138) (0.0425, 0.139861431257646) (0.045, 0.127414041125208) (0.0475, 0.117283492689293) (0.05, 0.108409893119923) (0.0525, 0.100595829332934) (0.055, 0.0936967528888617) (0.0575, 0.0876535733387352) (0.06, 0.0822820844185727)
 };
\addlegendentry{2Y6M}

\addplot[dotted,black,mark=diamond,]
coordinates {
(0.035, 0.20968777970729) (0.0375, 0.191130295627199) (0.04, 0.169301466723749) (0.0425, 0.153118190817983) (0.045, 0.137399964127397) (0.0475, 0.1267831893710239) (0.05, 0.1174034353108259) (0.0525, 0.109580011792202) (0.055, 0.1026365022713009) (0.0575, 0.0963820160649003) (0.06, 0.090791975301562)
};
\addlegendentry{2Y9M}

\addplot[dotted,green,mark=x,]
coordinates {
(0.035, 0.2179090907792) (0.0375, 0.198390313153192) (0.04, 0.17359264030603) (0.0425, 0.153858656222261) (0.045, 0.135893948372128) (0.0475, 0.125564097924078) (0.05, 0.116364537643523) (0.0525, 0.109136566833777) (0.055, 0.102691711371184) (0.0575, 0.096708399013374) (0.06, 0.0913336198112621)
 };
\addlegendentry{3Y}

\addplot[dashdotted,blue,mark=square,] 
coordinates {
(0.035, 0.20760437321493) (0.0375, 0.188630277237479) (0.04, 0.165927919533029) (0.0425, 0.147845862560755) (0.045, 0.131513823044396) (0.0475, 0.120879768086885) (0.05, 0.111420660831905) (0.0525, 0.103699054787514) (0.055, 0.096813507304964) (0.0575, 0.0905366913095371) (0.06, 0.0849051095006769)
};
\addlegendentry{3Y3M}

\addplot[dashdotted,red,mark=triangle,] 
coordinates {
(0.035, 0.194093129606618) (0.0375, 0.175805516283564) (0.04, 0.1556650775523479) (0.0425, 0.1395801631705299) (0.045, 0.125207672552225) (0.0475, 0.11424646490273) (0.05, 0.104519461353688) (0.0525, 0.0962632222937178) (0.055, 0.0889134288973133) (0.0575, 0.082347705583813) (0.06, 0.0764746928844306)
};
\addlegendentry{3Y6M}

\addplot[dashdotted,black,mark=diamond,]
coordinates {
(0.035, 0.168225813624404) (0.0375, 0.151690110957237) (0.04, 0.13548948735649) (0.0425, 0.122474051384803) (0.045, 0.111021001511994) (0.0475, 0.100307905901036) (0.05, 0.0908387452430459) (0.0525, 0.0824885550601205) (0.055, 0.075086381832756) (0.0575, 0.068621710116528) (0.06, 0.0628716762232644)
};
\addlegendentry{3Y9M}

\addplot[dashdotted,green,mark=x,]
coordinates {
(0.035, 0.142709305324433) (0.0375, 0.127903138210658) (0.04, 0.115580641152897) (0.0425, 0.105562580551698) (0.045, 0.096939310340912) (0.0475, 0.0864931073104497) (0.05, 0.0773140687450959) (0.0525, 0.0689398432388197) (0.055, 0.0615710474329548) (0.0575, 0.0552920753008763) (0.06, 0.0497575106216909)
};
\addlegendentry{4Y}

\addlegendimage{blue,mark=square}
\addlegendentry{4Y3M}

\addlegendimage{red,mark=triangle}
\addlegendentry{4Y6M}

\addlegendimage{black,mark=diamond}
\addlegendentry{4Y9M}

\addlegendimage{green,mark=x}
\addlegendentry{5Y}

\nextgroupplot[ xlabel={Strike price}]
        \addplot[blue,mark=square] 
 coordinates {       (0.035, 0.167060723324864) (0.0375, 0.151163367517036) (0.04, 0.136344841822182) (0.0425, 0.124258500932036) (0.045, 0.113593024768348) (0.0475, 0.102379888852465) (0.05, 0.0924016175432752) (0.0525, 0.083395203877002) (0.055, 0.0753724398217675) (0.0575, 0.0683925747207483) (0.06, 0.0621672285527745)
 };

\addplot[red,mark=triangle] 
coordinates {
(0.035, 0.191245945378995) (0.0375, 0.174373795626683) (0.04, 0.157084591567545) (0.0425, 0.142951148146641) (0.045, 0.130233683241115) (0.0475, 0.118397464774662) (0.05, 0.1077490759881699) (0.0525, 0.0982551902365554) (0.055, 0.0897092439008801) (0.0575, 0.0821315701942771) (0.06, 0.075303773486371)
 };

\addplot[black,mark=diamond]
coordinates {
(0.035, 0.218899349686139) (0.0375, 0.200822194430432) (0.04, 0.180765595252751) (0.0425, 0.164344499870033) (0.045, 0.14933741861967) (0.0475, 0.136773370261313) (0.05, 0.125363338184545) (0.0525, 0.115326459523536) (0.055, 0.106210770436498) (0.0575, 0.0979861396143332) (0.06, 0.0905099591077407)
};

\addplot[green,mark=x] 
coordinates {
(0.035, 0.243155277281691) (0.0375, 0.224219764606852) (0.04, 0.20172530694552) (0.0425, 0.183289200003796) (0.045, 0.1662245653445239) (0.0475, 0.153225355829318) (0.05, 0.141322286097815) (0.0525, 0.131003337210971) (0.055, 0.121558510770731) (0.0575, 0.1128972323639379) (0.06, 0.104963099675497)
};

\end{groupplot}
\node[align=center,anchor=north] at ([yshift=-12mm]my plots c1r1.south) {(a) Short-term maturities data set.};
\node[align=center,anchor=north] at ([yshift=-12mm]my plots c2r1.south) {(b) Long-term maturities data set.};
\path ([yshift=-12mm]my plots c1r1.south east) -- node[below,anchor=north]{\ref{CommonLegend}} ([yshift=-24mm]my plots c2r1.south west);
\end{tikzpicture}
\end{center}
\caption{Market caplet prices on the entire data set.}
\label{fig:Caplet_price_mkt}
\end{figure}

\begin{figure}[h!]
\centering\begin{tikzpicture} 
\pgfplotstableread[col sep=comma,]{MaturityStrikes.csv} {\prices}

\begin{axis}[
 xlabel={Maturity date (years)},
width=11cm,height=8cm,
legend pos=outer north east,
legend style={fill=white,fill opacity=1,text opacity=1},SimStyle]

\addplot[color=blue, mark=square, ] table [x={Expiry_dates}, y={3.5}] {\prices}; \addlegendentryexpanded{$K=3.5$}
\addplot[color=red, mark=triangle, ] table [x={Expiry_dates}, y={3.75}] {\prices}; \addlegendentryexpanded{$K=3.75$}
\addplot[color=black, mark=diamond, ] table [x={Expiry_dates}, y={4.0}] {\prices}; \addlegendentryexpanded{$K=4.0$}
\addplot[color=green, mark=x, ] table [x={Expiry_dates}, y={4.25}] {\prices}; \addlegendentryexpanded{$K=4.25$}
\addplot[color=blue, mark=square, dashed] table [x={Expiry_dates}, y={4.5}] {\prices}; \addlegendentryexpanded{$K=4.5$}
\addplot[color=red, mark=triangle, dashed] table [x={Expiry_dates}, y={4.75}] {\prices}; \addlegendentryexpanded{$K=4.75$}
\addplot[color=black, mark=diamond, dashed] table [x={Expiry_dates}, y={5.0}] {\prices}; \addlegendentryexpanded{$K=5.0$}
\addplot[color=green, mark=x, dashed] table [x={Expiry_dates}, y={5.25}] {\prices}; \addlegendentryexpanded{$K=5.25$}
\addplot[color=blue, mark=square, dotted] table [x={Expiry_dates}, y={5.5}] {\prices}; \addlegendentryexpanded{$K=5.5$}
\addplot[color=red, mark=triangle, dotted] table [x={Expiry_dates}, y={5.75}] {\prices}; \addlegendentryexpanded{$K=5.75$}
\addplot[color=black, mark=diamond, dotted] table [x={Expiry_dates}, y={6.0}] {\prices}; \addlegendentryexpanded{$K=6.0$}

\end{axis} 
 \end{tikzpicture}
 \vfill
 \caption{Caplet prices for each maturity date and a range of strikes on the entire data set.}
 \label{fig:mean_caplets}
\end{figure}

Calibration is performed by minimizing the following relative sum of squared errors \begin{equation}\label{eq:rel_error}
    \text{Error} = \sum_{i=1}^n \frac{1}{\text{Caplet}_{\text{Market}}^{i}} \left( \text{Caplet}_{\text{Market}}^{i} - \text{Caplet}_{\text{Model}}^{i}\right)^2 ,
\end{equation} where $\text{Caplet}_{\text{Market}}^{i}$ represents the market caplet price, $\text{Caplet}_{\text{Model}}^{i}$ is the caplet price given by the model, while $n$ is the number of caplets used in the calibration procedure. 

In what follows, we compare two approaches. Firstly, we calibrate against short-term maturities (the calibration set) and measure how well it performs for long-term maturities (the out-of sample set). Secondly, we calibrate using the entire dataset. The relevant sums of squared errors for the proposed model are presented in Tables \ref{table:SSE_all} and \ref{table:Rel_SSE_all}, alongside the Black, Bachelier, and Vasi\v{c}ek models. Parameter values appear in Table \ref{table:param_caplet_cali}.

\begin{table}
\tbl{Sum of squared errors obtained by different models on different data sets.}
{\begin{tabular}{|l| S[table-format=1.4e2] |S[table-format=1.5]| S[table-format=1.5] | }
\hline
Model &  \multicolumn{1}{c|}{Short-term maturities} & \multicolumn{1}{c|}{Long-term maturities} & \multicolumn{1}{c|}{Entire dataset} \\
 &  \multicolumn{1}{c|}{(calibrated)} & \multicolumn{1}{c|}{(out of sample)} & \multicolumn{1}{c|}{(calibrated)} \\\hline
\hline
Bachelier & 0.17405 & 0.02045 & 0.19252\\
\hline
Black & 0.18805 & 0.04221 &  0.20739  \\ 
\hline
Vasi\v{c}ek  &  0.16774 &  0.03726  & 0.18847  \\
\hline
Proposed model &  0.10292  & 0.02442  & 0.12752 \\
\hline
\end{tabular}}
 \label{table:SSE_all}
\end{table}

\begin{table}
\tbl{Relative sum of squared errors obtained by different models on different data sets.}
{\begin{tabular}{|l| S[table-format=1.4e2] |S[table-format=1.5]| S[table-format=1.5] | }
\hline
Model &  \multicolumn{1}{c|}{Short-term maturities} & \multicolumn{1}{c|}{Long-term maturities} & \multicolumn{1}{c|}{Entire dataset} \\
 &  \multicolumn{1}{c|}{(calibrated)} & \multicolumn{1}{c|}{(out of sample)} & \multicolumn{1}{c|}{(calibrated)} \\\hline
\hline
Bachelier & 2.11227 & 0.15997 & 2.26737 \\
\hline
Black &  2.39275 & 0.25876 &  2.61056   \\ 
\hline
Vasi\v{c}ek  & 1.92724 & 0.25374  & 2.09293  \\
\hline
Proposed model &  1.25596  & 0.18585  & 1.33631 \\
\hline
\end{tabular}}
 \label{table:Rel_SSE_all}
\end{table}

\begin{table}
\tbl{Calibrated parameters of the proposed model on different data sets.}
{\begin{tabular}{|c| S[table-format=2.5] |S[table-format=2.5]| S[table-format=2.5] | }
\hline
Parameter & \multicolumn{1}{c|}{Short-term maturities} & \multicolumn{1}{c|}{Entire data set} \\ \hline
\hline
$b$ &  -4.92594 $\cdot 10^{3}$   & -2.63030 $\cdot 10^{4}$ \\
\hline
$c_1$ & -7.94774 $\cdot 10^{2}$  & -4.40172 $\cdot 10^{3}$  \\ 
\hline
$\sigma$ & 2.92825 $\cdot 10^{2}$ & 1.51657 $\cdot 10^{3}$   \\
\hline
$\tau_1$ & 1.87482 &  1.64115  \\
\hline
\end{tabular}}
 \label{table:param_caplet_cali}
\end{table}

We first consider the performance on short-term maturities when calibrated against short-term maturities only (i.e. within the calibration set). It is clear from Tables \ref{table:SSE_all} and \ref{table:Rel_SSE_all}  that the proposed model obtains a significantly lower sum of squared errors compared to all other models (which perform similarly). The calibrated parameters (from Table \ref{table:param_caplet_cali}) show that the parameter $\tau_1$ is indeed near 2. The sum of squared errors for different maturities (in $\log_{10}$ terms) is presented in Figure \ref{fig:caplets_sse}. Overall, the proposed model achieves the smallest sum of squared errors for maturities, smaller than the calibrated $\tau_1=1.874 82$, except at 3M and 1Y9M. One can consider that the model can capture the dynamics shown in Figure \ref{fig:mean_caplets}. In other words, the model behaves differently depending on whether it has a short or long maturity. Figure \ref{fig:mean_in_sample} shows evidence of this, where we plot the mean of caplet prices per maturity date. We can observe that the model can adapt to the dynamics of caplet prices with maturities less than 2 years and the dynamics of caplet prices with maturities bigger than 2 years.

\begin{figure}[h!]
\centering\begin{tikzpicture} 
 \begin{axis}[
 xlabel={Maturity date (years)},
width=11cm,height=8cm,
legend pos=south east,
legend style={fill=white,fill opacity=1,text opacity=1},SimStyle]
\addplot[
  color=blue,
 mark=square,
]
 coordinates {
(0.252055, -7.427800622999848) (0.50411, -2.5754630514262242) (0.756164, -1.9882075434325277) (1.00274, -1.7139098901219667) (1.252055, -1.7991547550864968) (1.50411, -1.9873957339139854) (1.756164, -2.1238379155978975) (2.00274, -1.9357581209288386) (2.252055, -1.8917240032240925) (2.50411, -1.8327204400609003) (2.761644, -1.7365257827360168) (3.00274, -1.6920344192877304) (3.252055, -1.9866393560642006) (3.509589, -2.395434584954831) (3.761644, -2.387677939764917) (4.008219, -1.924993540692973)
};
\addlegendentry{Bachelier}

\addplot[
  color=red,
 mark=triangle,
]
 coordinates {
(0.252055, -7.2750696861258985) (0.50411, -2.231565786175953) (0.756164, -1.8028131931715332) (1.00274, -1.6330138695374605) (1.252055, -1.827090359186681) (1.50411, -2.1953988371562416) (1.756164, -2.5860377828404055) (2.00274, -2.089160649164806) (2.252055, -1.9329639928256428) (2.50411, -1.7756835573115113) (2.761644, -1.6048294239164083) (3.00274, -1.513395239185824) (3.252055, -1.741630413514029) (3.509589, -2.1384656162875832) (3.761644, -3.1703615934764784) (4.008219, -2.9418409233596283)
};
\addlegendentry{Black}

\addplot[
  color=black,
 mark=diamond,
]
 coordinates {
(0.252055, -7.430110822465058) (0.50411, -2.873125171485622) (0.756164, -2.206118800686241) (1.00274, -1.890557728016149) (1.252055, -1.9824179687904955) (1.50411, -2.145302258944985) (1.756164, -2.1496561461285175) (2.00274, -1.8761805633103061) (2.252055, -1.8413553062693742) (2.50411, -1.810696489920597) (2.761644, -1.7592696726689612) (3.00274, -1.7606221768806352) (3.252055, -2.156628857265302) (3.509589, -2.5169373954135956) (3.761644, -2.044156250854705) (4.008219, -1.590281693521361)
};
\addlegendentry{Vasi\v{c}ek}

\addplot[
 color=green,
 mark=x,
]
 coordinates {
(0.252055, -4.231992435516582) (0.50411, -3.224840857496841) (0.756164, -2.549030443589693) (1.00274, -2.2041864483607823) (1.252055, -2.3891413469855856) (1.50411, -2.2993637163052303) (1.756164, -1.8791025813918212) (2.00274, -2.018772380936854) (2.252055, -1.9905102493491658) (2.50411, -2.195568421486009) (2.761644, -2.165669007058113) (3.00274, -2.2711839263997953) (3.252055, -2.428170377684556) (3.509589, -2.518751941635426) (3.761644, -2.152352394793738) (4.008219, -1.7273889274147083)
};
\addlegendentry{Proposed model}
\end{axis} 
 \end{tikzpicture}
 \vfill
 \caption{ Sum of squared errors in $\log_{10}$ scale for each maturity date obtained by the different models on the short-term maturities data set.}
 \label{fig:caplets_sse}
\end{figure}

\begin{figure}[h!]
\centering\begin{tikzpicture} 
 \begin{axis}[
 xlabel={Maturity date (years)},
width=11cm,height=8cm,
legend pos=north east,
legend style={fill=white,fill opacity=1,text opacity=1},SimStyle]
\addplot[
  color=blue,
 mark=square,
]
 coordinates {
(0.252055, 0.23154387000581855) (0.50411, 0.12956010744343818) (0.756164, 0.11364634149087739) (1.00274, 0.09776381827187518) (1.252055, 0.08727877966734267) (1.50411, 0.08604689007170468) (1.756164, 0.0852858098339026) (2.00274, 0.08425266966131095) (2.252055, 0.08824247813360012) (2.50411, 0.09284574221973113) (2.761644, 0.09872584977163483) (3.00274, 0.09642525295684035) (3.252055, 0.10364887970376181) (3.509589, 0.11131753827826026) (3.761644, 0.11333173907009772) (4.008219, 0.1151250044659549)
};
\addplot[
  color=red,
 mark=triangle,
]
 coordinates {
(0.252055, 0.23155973525026793) (0.50411, 0.13739617331604334) (0.756164, 0.12279335105798914) (1.00274, 0.10483433708432424) (1.252055, 0.0908431256220494) (1.50411, 0.08705941432926595) (1.756164, 0.08363758599268474) (2.00274, 0.08024551410219753) (2.252055, 0.08239034850654922) (2.50411, 0.08523654252530188) (2.761644, 0.08926175079100553) (3.00274, 0.0860554152457947) (3.252055, 0.09156755486918998) (3.509589, 0.09746331486702134) (3.761644, 0.09845678654746984) (4.008219, 0.09940625356784713)
};
\addplot[
  color=black,
 mark=diamond,
]
 coordinates {
(0.252055, 0.23154081328408144) (0.50411, 0.12486601613385781) (0.756164, 0.10619897369559021) (1.00274, 0.08944931714111753) (1.252055, 0.07905269634300584) (1.50411, 0.07841337757022186) (1.756164, 0.07875417879133007) (2.00274, 0.07920758408471737) (2.252055, 0.08485912902073936) (2.50411, 0.09150790327420665) (2.761644, 0.09988867793878355) (3.00274, 0.09990889674840862) (3.252055, 0.11024442421326622) (3.509589, 0.12164350280307472) (3.761644, 0.12704856899761302) (4.008219, 0.132289550116026)
};
\addplot[
  color=green,
 mark=x,
]
 coordinates {
(0.252055, 0.23276179249704665) (0.50411, 0.1203078358138811) (0.756164, 0.09683626304122232) (1.00274, 0.07771090508960442) (1.252055, 0.06343296000944774) (1.50411, 0.05858014817863193) (1.756164, 0.055662225355358606) (2.00274, 0.09499732210659312) (2.252055, 0.1229238019592272) (2.50411, 0.12166758120664532) (2.761644, 0.11996567836336554) (3.00274, 0.12082812703856147) (3.252055, 0.12060951929187955) (3.509589, 0.12027021757952294) (3.761644, 0.12280009711977544) (4.008219, 0.12467345814114278)
};
\addplot[
  color=purple,
 mark=o,
]
 coordinates {
(0.252055, 0.2316042559220683) (0.50411, 0.11847556704340856) (0.756164, 0.08961436797021506) (1.00274, 0.06414223229958883) (1.252055, 0.05971835358606192) (1.50411, 0.07025161872451174) (1.756164, 0.08542412540361217) (2.00274, 0.10306349571731176) (2.252055, 0.11323094135108808) (2.50411, 0.12391710226610116) (2.761644, 0.13674680246503937) (3.00274, 0.13831305285727358) (3.252055, 0.13179791340109737) (3.509589, 0.12301059409827068) (3.761644, 0.10628322174651625) (4.008219, 0.08982387529368048)
};
\legend{Bachelier,Black, Vasi\v{c}ek, Proposed model,Market}
\end{axis} 
 \end{tikzpicture}
 \caption{Mean of caplet prices for each maturity date on the short-term maturities data set. }
 \label{fig:mean_in_sample}
\end{figure}

Turning to the out-of-sample set, the longer-term maturities, it can be seen from Table \ref{table:SSE_all} that the proposed model exhibits a sum of squared errors, which aligns with the other benchmark models. However, the proposed model performs slightly worse than the Bachelier model. Figure \ref{fig:caplets_sse_out} illustrates the $\log_{10}$ sum of squared errors for each maturity date. Our model does not perform best for any particular expiry date.

\begin{figure}[h!]
\centering\begin{tikzpicture} 
 \begin{axis}[
 xlabel={Maturity date (years)},
width=11cm,height=8cm,
legend pos = south east,
legend style={fill=white,fill opacity=1,text opacity=1},SimStyle]
\addplot[
 color=blue,
 mark=square,
]
 coordinates {
(4.257534, -2.2556025708518685) (4.506849, -2.6416955709737) (4.758904, -2.4799182951188823) (5.008219, -2.031199023013279)
};
\addplot[
 color=red,
 mark=triangle,
]
 coordinates {
(4.257534, -3.8919931252851945) (4.506849, -2.4602744484316474) (4.758904, -1.9135727235388442) (5.008219, -1.5780761227997597)
};
\addplot[
 color=black,
 mark=diamond,
]
 coordinates {
(4.257534, -1.74043653789572) (4.506849, -1.9536897833520834) (4.758904, -2.2355515744419314) (5.008219, -2.66849866387802)
};
\addplot[
 color=green,
 mark=x,
]
 coordinates {
(4.257534, -2.095506306509955) (4.506849, -2.601979421478236) (4.758904, -2.458064043150065) (5.008219, -1.9824584567034589)
};
\legend{Bachelier,Black, Vasi\v{c}ek,Proposed model}
\end{axis} 
 \end{tikzpicture}
  \caption{Sum of squared errors in $\log_{10}$ scale on long-term maturities by models calibrated to short-term maturities.}
 \label{fig:caplets_sse_out}
\end{figure}

Finally, we calibrate all models using the entire dataset. The calibration minimizes the relative sum of squared errors in \eqref{eq:rel_error}. The sum of squared errors for all models are presented in Tables \ref{table:SSE_all} and \ref{table:Rel_SSE_all}. Our results indicate that the proposed model achieves the lowest error when calibrated using the entire dataset. Table \ref{table:SSE_ExpDateCali_total} provides the parameter values obtained from the entire dataset. According to Proposition \ref{prop:stationary_distr}, the model exhibits a non-degenerate limiting distribution for these parameter values. Furthermore, the calibrated value of the delay parameter $\tau_1$ is again near to the value $2$. Figure \ref{fig:caplets_sse_all} shows that for maturities smaller than the calibrated delay parameter $\tau_1 =1.64115$, the proposed model obtains the smallest error except when for the maturity 3M. This can also be seen in Table \ref{table:SSE_ExpDateCali_total}. Figure \ref{fig:mean_in_sample_Total} shows the mean of caplet prices per maturity date; notice that the model differentiates between the caplets with maturities less than $\tau_1 = 1.641 15$ and with maturities bigger than  $\tau_1 = 1.641 15$.

Comparing the estimated delay parameters obtained in Section \ref{sect:estimation} with those obtained in the calibration procedure. One can observe that the estimated delays in Section \ref{sect:estimation} are smaller than those obtained from caplet calibration. This means the market measure differs from the risk-neutral measure. As we showed in Appendix \ref{sec:change-of-measure}, the interest rate process can exhibit different delays under the risk-neutral and physical measures. This is consistent with the empirical results that we obtained. The delay parameter selected according to the caplets is $\tau_1 = 1.874$ while the delay parameter with the highest value in the periodogram is $\tau_1 = 0.019$. According to \cite[Section 4.4]{smith2011introduction}, the deterministic delay differential equation \eqref{eq:deter}  will behave similarly to an ordinary differential equation when the delay parameter is small. This implies that the value of caplets for a small delay parameter $\tau_1$, will be similar to the Vasi\v{c}ek model.

\begin{table}
\tbl{Sum of squared errors in $\log_{10}$ scale  on the entire data set.}
{\begin{tabular}{|l| S[table-format=2.4] | S[table-format=2.4] | S[table-format=2.4]  | S[table-format=2.4] | S[table-format=2.4] | }
\hline
Maturity & \multicolumn{1}{c|}{Bachelier} & \multicolumn{1}{c|}{Black} & \multicolumn{1}{c|}{Vasi\v{c}ek} & \multicolumn{1}{c|}{Proposed model} \\
date &  &  &  &   \\\hline
\hline
3M&-7.4268&-7.2119&-7.4304&-6.6549\\
\hline
6M&-2.5526&-2.1719&-2.8177&-3.2824\\
\hline
9M&-1.9692&-1.7489&-2.1693&-2.6055\\
\hline
1Y&-1.6964&-1.5812&-1.8644&-2.2631\\
\hline
1Y3M&-1.7779&-1.7586&-1.9606&-2.4405\\
\hline
1Y6M&-1.9634&-2.0976&-2.1341&-2.2727\\
\hline
1Y9M&-2.1136&-2.5634&-2.1497&-2.046\\
\hline
2Y&-1.9499&-2.1873&-1.8699&-1.7732\\
\hline
2Y3M&-1.9131&-2.0362&-1.8232&-2.0085\\
\hline
2Y6M&-1.8598&-1.8725&-1.7761&-2.1966\\
\hline
2Y9M&-1.7668&-1.6903&-1.7052&-2.1518\\
\hline
3Y&-1.7228&-1.5872&-1.6882&-2.2383\\
\hline
3Y3M&-2.029&-1.8457&-2.0288&-2.3739\\
\hline
3Y6M&-2.4335&-2.3213&-2.4766&-2.5161\\
\hline
3Y9M&-2.3426&-4.0151&-2.2419&-2.1836\\
\hline
4Y&-1.8822&-2.6076&-1.7529&-1.779\\
\hline
4Y3M&-2.2003&-4.3033&-1.9724&-2.1693\\
\hline
4Y6M&-2.611&-2.7514&-2.2956&-2.6517\\
\hline
4Y9M&-2.5517&-2.0626&-2.6992&-2.3543\\
\hline
5Y&-2.0931&-1.6779&-2.7319&-1.9441\\

\hline
\end{tabular}}
 \label{table:SSE_ExpDateCali_total}
\end{table}

\begin{figure}[h!]
\centering\begin{tikzpicture} 
 \begin{axis}[
 xlabel={Maturity date (in years)},
width=11cm,height=8cm,
legend pos = south east,
legend style={fill=white,fill opacity=1,text opacity=1},SimStyle]
\addplot[
 color=blue,
 mark=square,
]
 coordinates {
(0.252055, -7.426770782695635) (0.50411, -2.552616841856077) (0.756164, -1.9691854865965273) (1.00274, -1.6963758061678975) (1.252055, -1.777923149213694) (1.50411, -1.9634236881506595) (1.756164, -2.1136190537516204) (2.00274, -1.9499302789696096) (2.252055, -1.9130686071949812) (2.50411, -1.8598154249095256) (2.761644, -1.7668157817342969) (3.00274, -1.7228185343507085) (3.252055, -2.0290035008118545) (3.509589, -2.4335103515247347) (3.761644, -2.342646715547975) (4.008219, -1.8821697814954956) (4.257534, -2.2003222903520445) (4.506849, -2.610979609416441) (4.758904, -2.551654955703708) (5.008219, -2.09305301013567)
};
\addplot[
 color=red,
 mark=triangle,
]
 coordinates {
(0.252055, -7.211861356918725) (0.50411, -2.1718946005283915) (0.756164, -1.748902757499878) (1.00274, -1.581221065046476) (1.252055, -1.758628931295345) (1.50411, -2.097559618136721) (1.756164, -2.5633946156539866) (2.00274, -2.187311665913043) (2.252055, -2.0361701876230454) (2.50411, -1.8725135344070092) (2.761644, -1.6902785365822075) (3.00274, -1.587227148713081) (3.252055, -1.8457121076449061) (3.509589, -2.3212581918934063) (3.761644, -4.015136299427025) (4.008219, -2.6076094532251397) (4.257534, -4.303346603407228) (4.506849, -2.7514466770431922) (4.758904, -2.0626436313319974) (5.008219, -1.677918940133642)
};
\addplot[
 color=black,
 mark=diamond,
]
 coordinates {
(0.252055, -7.430386202152834) (0.50411, -2.817746548877908) (0.756164, -2.169288616472981) (1.00274, -1.8643765334583626) (1.252055, -1.9605534429374811) (1.50411, -2.134119961903651) (1.756164, -2.149726642845529) (2.00274, -1.869893927616126) (2.252055, -1.8231641752268821) (2.50411, -1.7761250505653672) (2.761644, -1.7052320339682256) (3.00274, -1.6881772196532179) (3.252055, -2.02880887398938) (3.509589, -2.476576540046345) (3.761644, -2.2418904404032305) (4.008219, -1.7528796660340984) (4.257534, -1.9723903425436822) (4.506849, -2.295552623258762) (4.758904, -2.6992358377222225) (5.008219, -2.731891039575178)
};
\addplot[
 color=green,
 mark=x,
]
 coordinates {
(0.252055, -6.654874278678436) (0.50411, -3.2824129418338295) (0.756164, -2.6055134702651794) (1.00274, -2.2631334897350905) (1.252055, -2.4404961367785813) (1.50411, -2.2727358537429234) (1.756164, -2.046036321263238) (2.00274, -1.7732324398140695) (2.252055, -2.0085231394527727) (2.50411, -2.1966347769185) (2.761644, -2.1518110624990707) (3.00274, -2.2382549412238437) (3.252055, -2.3739122333753455) (3.509589, -2.516124343848065) (3.761644, -2.1836038204967667) (4.008219, -1.778954909733854) (4.257534, -2.16934130099718) (4.506849, -2.65170339841845) (4.758904, -2.3543311135692337) (5.008219, -1.9441100535006466)
};
\legend{Bachelier,Black, Vasi\v{c}ek,Proposed model}
\end{axis} 
 \end{tikzpicture}
\caption{Sum of squared errors in $\log_{10}$ scale for each maturity date obtained by the different models on the whole data set.}
 \label{fig:caplets_sse_all}
\end{figure}

\begin{figure}[h!]
\centering\begin{tikzpicture} 
 \begin{axis}[
 xlabel={Maturity date (years)},
width=11cm,height=8cm,
legend pos=north east,
legend style={fill=white,fill opacity=1,text opacity=1},SimStyle]
\addplot[
  color=blue,
 mark=square,
]
 coordinates {
(0.252055, 0.23154425014175867) (0.50411, 0.12998079947742133) (0.756164, 0.11437711557550258) (1.00274, 0.09867870833122715) (1.252055, 0.08832828866068401) (1.50411, 0.08721833704223753) (1.756164, 0.08655124880875857) (2.00274, 0.08557522460493487) (2.252055, 0.08966557356190319) (2.50411, 0.094368351756393) (2.761644, 0.10036310514531337) (3.00274, 0.09803398525286143) (3.252055, 0.10538179164285975) (3.509589, 0.11317845312681611) (3.761644, 0.11522249108198443) (4.008219, 0.11703862126685306) (4.257534, 0.12267170750359123) (4.506849, 0.12712417493703365) (4.758904, 0.13309775955979306) (5.008219, 0.13634043786315875)
};
\addplot[
 color=red,
 mark=triangle,
]
 coordinates {
(0.252055, 0.2315639447241641) (0.50411, 0.138986285929613) (0.756164, 0.12531851733122587) (1.00274, 0.10783289865271851) (1.252055, 0.09413110907437451) (1.50411, 0.09061400741941585) (1.756164, 0.08736339761298863) (2.00274, 0.08403703957178948) (2.252055, 0.08639291705881934) (2.50411, 0.08944875479150949) (2.761644, 0.09372333464385528) (3.00274, 0.09038113806712705) (3.252055, 0.09617654223237229) (3.509589, 0.10236297619521442) (3.761644, 0.10338892450403186) (4.008219, 0.10435841736527798) (4.257534, 0.10887076345491759) (4.506849, 0.11243475369568001) (4.758904, 0.11747066621372522) (5.008219, 0.12023144435994247)
};
\addplot[
  color=black,
 mark=diamond,
]
 coordinates {
(0.252055, 0.23154115952626392) (0.50411, 0.12564864385358765) (0.756164, 0.10735626159068311) (1.00274, 0.09059421691770662) (1.252055, 0.07997684001246246) (1.50411, 0.0789822566360091) (1.756164, 0.0788432252674187) (2.00274, 0.07873595579403954) (2.252055, 0.08370804158666773) (2.50411, 0.08954510564655506) (2.761644, 0.09693542577458553) (3.00274, 0.09618966639902636) (3.252055, 0.10529376511089204) (3.509589, 0.11523229032227902) (3.761644, 0.119393268959494) (4.008219, 0.12336435524260887) (4.257534, 0.13159426876236138) (4.506849, 0.1387320964763681) (4.758904, 0.14777375771471976) (5.008219, 0.15387318280605983)
};
\addplot[
 color=green,
 mark=x,
]
 coordinates {
(0.252055, 0.2316030296182204) (0.50411, 0.1194675258315575) (0.756164, 0.09543133722754625) (1.00274, 0.07582560519260874) (1.252055, 0.06130013165198564) (1.50411, 0.056244547965930675) (1.756164, 0.09274917475093906) (2.00274, 0.12427523426594343) (2.252055, 0.12147168824000604) (2.50411, 0.12003140354936746) (2.761644, 0.11926040458161637) (3.00274, 0.11964099511782819) (3.252055, 0.11838241757169649) (3.509589, 0.1192171137182965) (3.761644, 0.12160563599964734) (4.008219, 0.12205664289055393) (4.257534, 0.12398527010617438) (4.506849, 0.12523047362613582) (4.758904, 0.1271309189260212) (5.008219, 0.13084319755059)
};
\addplot[
 color=purple,
 mark=o,
]
 coordinates {
(0.252055, 0.2316042559220683) (0.50411, 0.11847556704340852) (0.756164, 0.08961436797021508) (1.00274, 0.06414223229958883) (1.252055, 0.05971835358606193) (1.50411, 0.07025161872451174) (1.756164, 0.08542412540361217) (2.00274, 0.10306349571731176) (2.252055, 0.11323094135108809) (2.50411, 0.12391710226610114) (2.761644, 0.13674680246503937) (3.00274, 0.13831305285727355) (3.252055, 0.1317979134010974) (3.509589, 0.1230105940982707) (3.761644, 0.10628322174651626) (4.008219, 0.08982387529368048) (4.257534, 0.10695721924840895) (4.506849, 0.12431231659471771) (4.758904, 0.14421264499881736) (5.008219, 0.1621439941936957)
};

\legend{Bachelier,Black, Vasi\v{c}ek, Proposed model,Market}
\end{axis} 
 \end{tikzpicture}
 \caption{Mean of caplet prices for each maturity date on the entire data set. }
 \label{fig:mean_in_sample_Total}
\end{figure}

\section{Conclusion}

In this paper, we obtained an analytical formula for the price of a bond when the short rate satisfies a stochastic delay differential equation. The model presented is a delayed version of the well-known Vasi\v{c}ek model. To our knowledge, this is the first time an analytical formula has been derived for the zero coupon bond under this model. We also give an analytical formula for the strong solution of the stochastic delay differential equation \eqref{eq:r}. In addition, we showed that the short rate follows a normal distribution, and has a limiting distribution under certain conditions. This last result has applications outside of the fixed-income securities context. For example, the valuation of weather derivatives uses an Ornstein-Uhlenbeck process with a periodic function in the drift to model temperature \citep{esunge2020weather}. It might therefore be possible to use equation \eqref{eq:r} to model the temperature since the presence of the delay parameter allows the model to capture the past dependencies that appear on the temperature. Additionally, we derived an analytical formula for the instantaneous forward rate, enabling us to apply the deterministic shift extension method  and achieve perfect calibration of the yield curve. Furthermore, we obtained an analytical formula for caplets on overnight risk-free rates. Numerical experiments were also conducted, demonstrating that the proposed model performs favorably compared to various benchmark models.

The techniques presented here can be applied to develop new model variations. For instance, a second delayed Ornstein-Uhlenbeck process can be introduced to create a delayed version of the two-factor Hull-White model \citep{hull1994numerical2}. 

In this framework, we assume that the noise in equation \eqref{eq:r} follows a Brownian motion, although a Lévy process could replace it. This substitution would result in a Lévy-driven Ornstein-Uhlenbeck process with delay. Lévy-driven Ornstein-Uhlenbeck processes are widely used in various financial applications, including commodities \citep{li2014time}, energy derivatives \citep{benth2007non}, and volatility derivatives \citep{Habtemicael2016}. 

Recently, Guinea Juli{\'a} and Caro-Carretero \cite{GuineaJulia2024} utilized the methods described here to construct a delayed version of the Barndorff-Nielsen and Shephard model \citep{nicolato2003option,Barndorff-Nielsen_Shephard2001a,Barndorff-Nielsen_Shephard2001b}. The results presented in this work may open new avenues for research in the field.

\bibliographystyle{tfs}
\bibliography{interacttfssample}

\appendix
\section{Change of measure}\label{sec:change-of-measure}

Instead of starting from a pre-selected risk-neutral measure $\mathds{Q}$ as we have done in this paper, it is possible to model the short rate under the real-world measure $\mathds{P}$, and then make a structure-preserving change of measure, the goal being the stochastic differential equation \eqref{eq:r}. To this end, let us assume that the short rate satisfies the stochastic differential equation
\begin{equation}
dr_t = \left(\alpha(t) + \beta r_t + \sum_{j=1}^J \gamma_j r_{t-\eta_j}\right)dt + \sigma(t) dB_t \label{eq:rP}
\end{equation}
under $\mathds{P}$, where $B=(B_t)$ is a Brownian motion under $\mathds{P}$ and where $\beta,\gamma_1,\ldots,\gamma_J\in\mathds{R}$, $\eta_J>\ldots>\eta_1>0$, and $\alpha:[0,\infty]\to \mathds{R}$, $\sigma:[0,\infty]\to (0,\infty)$ are two continuous functions with
\[ \sup_{s\in[0,\infty)} \frac{1}{\sigma^2(s)} <\infty.\] We take the parameters $a(.),b$ and $c_1,\ldots,c_N$ from \eqref{eq:r}.

For given (possibly large) time horizon $H>0$, define the processes $\lambda=(\lambda_t)_{t\in[0,H]}$ and $Z=(Z_t)_{t\in[0,H]}$ as
\begin{align}
\lambda_t &=  \frac{1}{\sigma(t)}\left[ a(t) - \alpha(t) + (b -\beta) r_t  - \sum_{j=1}^J \gamma_j r_{t-\eta_j}  + \sum_{i=1}^N c_i r_{t-\tau_i} \right], \label{eq:lambda} \\
Z_t &= e^{-\int_0^t \lambda_s dB_s - \frac{1}{2} \int_0^t \lambda_s^2 ds}
\end{align}
for all $t\in[0,H]$. It is sufficient to show that $Z$ is a martingale, because then, by the Girsanov theorem there exists a measure $\mathds{Q}$ such that $Z_T = \frac{d\mathds{Q}}{d\mathds{P}}$ and a Brownian motion $W=(W_t)_{t=0}^H$ under $\mathds{Q}$ such that \eqref{eq:r} holds true for all $t\in[0,H]$. This achieves the objective.

\begin{proposition}
 If the function $\phi$ is bounded on $[-(\tau_N\vee \eta_J) ,0]$, then $Z$ is a martingale.
\end{proposition}

\begin{proof}
We use a result by Klebaner and Liptser \citep[Theorem 5.1]{klebaner2014stochastic}. In order to show that $Z$ is a martingale, we need to show that there exists a number $\xi\geq |r_0|$ such that 
\begin{equation}
    \max\left\{\lambda_t^2, \left(a(t) +\sum_{j=1}^N c_j r_{t-\tau_j} + br_t \right)^2 + \sigma^2(t)\left(1+ \lambda_t^2\right)\right\} \leq \xi \left( 1+ \sup_{s\in[0,t)} r_s^2\right)\label{eq:Cond3}
\end{equation} 
for all $t\in(0,H]$.

Define
\begin{align*}
 c_\phi & = \sup_{s\in[-(\tau_N\vee \eta_J),0]} |\phi(s)|,\\
 c_a & = \sup_{s\in[0,H]} (a(s)-\alpha(s))^2,\\
 c_{\sigma} & =  \sup_{s\in[0,H]} \frac{1}{\sigma^2(s)}.
\end{align*} Observe that the Cauchy–Schwarz inequality in $\mathds{R}^{J+N}$, gives us
\begin{align}
    &\left(\sum_{j=1}^J \gamma_j r_{t-\eta_j}  + \sum_{i=1}^N c_i r_{t-\tau_i}\right)^2  \leq (J+N)  \left(\sum_{j=1}^J  \gamma_j^2 r^2_{t-\eta_j}  + \sum_{i=1}^N c_i^2 r^2_{t-\tau_i}\right) \nonumber \\
    &= (J+N)  \left(\sum_{j=1}^J  \gamma_j^2 \phi^2(t-\eta_j) \mathds{1}_{[0,\eta_j]}(t) + \sum_{i=1}^N c_i^2 \phi^2(t-\tau_i) \mathds{1}_{[0,\tau_i]}(t)\right) \nonumber\\
    & + (J+N) \left(\sum_{j=1}^J  \gamma_j^2 r^2_{t-\eta_j}\mathds{1}_{(\eta_j,\infty]}(t)  + \sum_{i=1}^N c_i^2 r^2_{t-\tau_i} \mathds{1}_{(\tau_i,\infty]}(t)\right) \nonumber \\
    & \leq (J+N) c_\phi^2 \mathcal{M} + (J+N) \mathcal{M}   \sup_{s\in[0,t)} r_s^2,\label{eq:CS}
\end{align} where
\[ \mathcal{M} = \left( \sum_{j=1}^J  \gamma_j^2  + \sum_{i=1}^N c_i^2 \right).\]

The identity $(x+y)^2 \le 2x^2+2y^2$ for all $x,y\in\mathds{R}$ and the inequality \eqref{eq:CS}, alow us to obtain
\begin{align*}
    \lambda_t^2 
     & \leq  \frac{2}{\sigma^2(t)}\left(2( a(t)-\alpha(t))^2 + (J+N) c_\phi^2 \mathcal{M}  +  2(b-\beta)^2 r_t^2 + (J+N) \mathcal{M}  \sup_{s\in[0,t)} r_s^2\right) \\
    &\leq \frac{2}{c_{\sigma}}\left(2 c_a +  (J+N) c_\phi^2 \mathcal{M} +  \left(2(b-\beta)^2 + (J+N) \mathcal{M} \right)\ \sup_{s\in[0,t)} r_s^2\right)\\
    & \leq c_\lambda\left(1 + \sup_{s\in[0,t)} r_s^2\right),
\end{align*}
where
\begin{align*}
 c_\lambda & = \frac{2}{c_{\sigma}}\max\left\{ 2 c_a +  (J+N) c_\phi^2 \mathcal{M}, 2(b-\beta)^2 + (J+N) \mathcal{M}\right\}.
\end{align*}
Similarly, 
\begin{align*}
 (a(t) + br_t + \sum_{j=1}^N c_j r_{t-\tau_j})^2 + \sigma^2(t) 
 &\leq 2\left(2a^2(t) + 2b^2r_t^2 + N \sum_{j=1}^N c^2_jr_{t-\tau_j}^2\right) + \sigma^2(t) \\
 &\leq 4\tilde{c}_a + 2N c_\phi^2 \sum_{j=1}^N c_j^2 + \tilde{c}_{\sigma} \\
 & + \left(4b^2 + 2N \sum_{j=1}^N c_j^2 \right)\sup_{s\in[0,t)} r_s^2,
\end{align*} where\begin{align*}
    \tilde{c}_a & = \sup_{s\in [0,t)} a^2(s),\\
    \tilde{c}_{\sigma} & = \sup_{s\in [0,t)} \sigma^2(s).
\end{align*} Hence, the required number $\xi$ is
\[
 \xi = \max\left\{|r_0|, c_\lambda, \max\left\{4\tilde{c}_a + 2N c_\phi^2 \sum_{j=1}^N c_j^2 + \tilde{c}_{\sigma}, 4b^2 + 2N \sum_{j=1}^N c_j^2\right\} + \tilde{c}_{\sigma}c_\lambda\right\}.
\]
\end{proof}

\end{document}